\documentclass[12pt,letterpaper]{article}

\usepackage{natbib}
\usepackage[ left=1in, top=1in, right=1in, bottom=1in]{geometry}
\usepackage{xcolor,graphicx,bm,bbm,colonequals,amsmath,amssymb,url}
\usepackage{array,tabularx,multirow}
\usepackage{enumitem,algpseudocode}
\usepackage[font={footnotesize}]{caption,subcaption}
\usepackage{color}

\usepackage{xr}
\makeatletter
\newcommand*{\addFileDependency}[1]{
  \typeout{(#1)}
  \@addtofilelist{#1}
  \IfFileExists{#1}{}{\typeout{No file #1.}}
}
\makeatother
\newcommand*{\myexternaldocument}[1]{%
    \externaldocument{#1}%
    \addFileDependency{#1.tex}%
    \addFileDependency{#1.aux}%
}
\myexternaldocument{supplement}

\usepackage[normalem]{ulem} 
\usepackage{bbm}
\newcommand{\indicat}{\mathbbm{1}}

\bibpunct[, ]{(}{)}{;}{a}{,}{,}

\usepackage{amsthm}
\newtheoremstyle{propstyle} 
    {3mm}                    
    {1mm}                    
    {\itshape}                   
    {}                           
    {\scshape}                   
    {.}                          
    {.5em}                       
    {}  
\theoremstyle{propstyle}
\newtheorem{prop}{Proposition}
\theoremstyle{propstyle}

\theoremstyle{propstyle}

\newsavebox\ideabox

\newcommand{\ba}{\mathbf{a}}

\newcommand{\bb}{\mathbf{b}}

\newcommand{\bs}{\mathbf{s}}

\newcommand{\bx}{\mathbf{x}}
\newcommand{\by}{\mathbf{y}}

\newcommand{\bS}{\mathbf{S}}
\newcommand{\bz}{\mathbf{z}}

\newcommand{\bW}{\mathbf{W}}

\newcommand{\bI}{\mathbf{I}}
\newcommand{\bD}{\mathbf{D}}
\newcommand{\bH}{\mathbf{H}}
\newcommand{\bU}{\mathbf{U}}
\newcommand{\bV}{\mathbf{V}}
\newcommand{\bK}{\mathbf{K}}

\newcommand{\bQ}{\mathbf{Q}}

\newcommand{\bC}{\mathbf{C}}
\newcommand{\bM}{\mathbf{M}}

\newcommand{\all}{\bullet}

\newcommand{\bfzero}{\mathbf{0}}
\newcommand{\bfone}{\mathbf{1}}

\newcommand{\bfmu}{\bm{\mu}}

\newcommand{\bftheta}{\bm{\theta}}

\newcommand{\bfSigma}{\bm{\Sigma}}

\DeclareMathOperator{\E}{\mathbb{E}}
\DeclareMathOperator{\var}{var}
\DeclareMathOperator{\cov}{cov}
\DeclareMathOperator{\diag}{diag}
\DeclareMathOperator{\chol}{chol}
\DeclareMathOperator{\rchol}{rchol}

\DeclareMathOperator{\rev}{rev}

\DeclareMathOperator{\blockdiag}{blockdiag}
\newcommand{\KL}{\textnormal{KL}}
\newcommand{\CKL}{\textnormal{CKL}}
\newcommand{\GP}{GP}

\newcommand{\normal}{\mathcal{N}}
\newcommand{\order}{\mathcal{O}}

\newcommand{\domain}{\mathbb{D}} 

\newcommand{\dens}{f}
\newcommand{\adens}{\widehat{f}}

\newcommand{\locs}{\mathcal{S}}

\DeclareMathOperator*{\argmin}{arg\,min}
\DeclareMathOperator*{\argmax}{arg\,max}

\usepackage{pifont}
\newcommand{\cmark}{\ding{51}}
\newcommand{\xmark}{\ding{55}}

\title{~ \\ \vspace{-15mm}
Vecchia approximations of Gaussian-process predictions}

\author{Matthias Katzfuss\thanks{Department of Statistics, Texas A\&M University} \thanks{Corresponding author: \texttt{katzfuss@gmail.com}} \and Joseph Guinness\thanks{Department of Statistics and Data Science, Cornell University}
\and Wenlong Gong\footnotemark[1]
\and Daniel Zilber\footnotemark[1]}

\date{}

\begin{document}

\maketitle

\begin{abstract}
Gaussian processes (GPs) are highly flexible function estimators used for geospatial analysis, nonparametric regression, and machine learning, but they are computationally infeasible for large datasets. Vecchia approximations of GPs have been used to enable fast evaluation of the likelihood for parameter inference. Here, we study Vecchia approximations of spatial predictions at observed and unobserved locations, including obtaining joint predictive distributions at large sets of locations. We consider a general Vecchia framework for GP predictions, which contains some novel and some existing special cases. We study the accuracy and computational properties of these approaches theoretically and numerically, proving that our new methods exhibit linear computational complexity in the total number of spatial locations. We show that certain choices within the framework can have a strong effect on uncertainty quantification and computational cost, which leads to specific recommendations on which methods are most suitable for various settings. We also apply our methods to a satellite dataset of chlorophyll fluorescence, showing that the new methods are faster or more accurate than existing methods, and reduce unrealistic artifacts in prediction maps.
\end{abstract}


\section{Introduction \label{sec:intro}}

Gaussian processes (GPs) are popular models for functions, time series, and spatial fields, with many application areas such as geospatial analysis \citep[e.g.,][]{Banerjee2004,Cressie2011}, nonparametric regression and machine learning \citep[e.g.,][]{Rasmussen2006}, the analysis of computer experiments \citep[e.g.,][]{Kennedy2001}, and Bayesian optimization of expensive functions \citep{Jones1998} and of the tuning parameters in neural networks \citep[e.g.,][]{Snoek2012}. Here, we focus on spatial prediction using GPs.
GPs are flexible, interpretable, allow natural probabilistic quantification of uncertainty, and are thus well-suited for big-data applications in principle.
However, direct application of GPs incurs computational cost that is cubic in the data size, which is too expensive for many modern datasets of interest.

To deal with this computational problem, numerous GP approximations or simplifying assumptions have been proposed. These include imposing sparsity on covariance matrices \citep{furrer2006covariance,kaufman2008covariance,Du2009}, sparsity on precision matrices \citep{rue2005gaussian,lindgren2011explicit,Nychka2012}, composite likelihoods \citep[e.g.,][]{curriero1999composite,Stein2004,Eidsvik2012}, and low-rank structure \citep[e.g.,][]{Higdon1998, Wikle1999, Quinonero-Candela2005, Banerjee2008, Cressie2008, Katzfuss2010,tzeng2018resolution}. While low-rank approaches are poorly suited for capturing fine-scale dependence, sparsity-based approaches can generally not guarantee linear scaling in the data or grid size, especially in higher dimensions. Local GP approximations \citep[e.g.,][]{Gramacy2015} are fast but do not scale well to joint predictions at many locations.

We focus on Vecchia approximations, which obtain a sparse Cholesky factor of the precision matrix by removing conditioning variables in a factorization of the joint density of the GP observations into a product of conditional distributions \citep{Vecchia1988}. This approach has become very popular for likelihood approximations for parameter inference \citep[e.g.,][]{Stein2004,Sun2016,Guinness2016a}.
For the typical setting of GP observations that include additive noise, \citet{Katzfuss2017a} consider a general Vecchia framework that applies the Vecchia approximation to a vector consisting of both the latent GP realizations and the noisy data. This framework contains many other popular GP approximations as special cases \citep[e.g.,][]{Snelson2007,Finley2009,Sang2011a,Datta2016,Katzfuss2015,Katzfuss2017b}.

Several authors have also proposed the use of Vecchia approximations for the important task of GP prediction, also referred to as kriging. The approach in \citet{Vecchia1992} for one-at-a-time Vecchia predictions has squared time complexity in the number of observations. \citet{Datta2016} and \citet{Finley2017} proposed Bayesian inference and prediction based on Vecchia-type approximations, which we will discuss and compare to in detail in the present paper. \citet{Guinness2016a} considered prediction using conditional expectation and uncertainty quantification using conditional simulations in a Vecchia approach based solely on conditioning on observed variables. This is relatively computationally cheap, but uncertainty measures contain random simulation error, and the observed conditioning might not provide accurate approximations in the presence of noise \citep[cf.][]{Katzfuss2017a}. Vecchia approximations have also been employed as preconditioners in iterative solvers that are used in prediction \citep{stroud2017bayesian}, but this approach is feasible only if prediction is desired at a small number of locations, or if additional approximations are made \citep{guinness2017spectral}. The multi-resolution approximation \citep{Katzfuss2015,Katzfuss2017b} and related approaches relying on domain partitioning \citep[e.g.,][]{Sang2011a,Zhang2019}, shown in \citet{Katzfuss2017a} to be special cases of Vecchia approximations, also provide fast GP prediction, but they can lead to artifacts along partition boundaries. We will discuss these connections and provide numerical comparisons.

Our article synthesizes and extends the literature on Vecchia approximations of spatial GP predictions, in particular the use of the general Vecchia approximation \citep{Katzfuss2017a} for marginal and joint predictive distributions. Extension of the general Vecchia framework to GP prediction was not considered in \citet{Katzfuss2017a} and requires consideration of complex issues, including how to order variables and choose conditioning sets to achieve accurate joint predictions at observed and unobserved locations, and how to guarantee fast computation of relevant summaries of the joint predictive distribution.
Here, we systematically study these issues with regard to the accuracy of the approximations and their computational burden. We introduce novel approaches within the framework, for which we can guarantee sparsity of the matrices necessary for inference, resulting in linear memory and time complexity in the number of data points and predictions for fixed conditioning-set size. Our framework enables systematic discussion, study, and comparison of our new methods and existing approaches, based on which we make specific recommendations about which methods are most suitable in various situations. Our framework is agnostic with respect to the inferential framework, allowing both frequentist and Bayesian inference on potential hyperparameters. We focus on the approximation of predictive distributions conditional on hyperparameters, to avoid confounding with the choice of hyperparameter priors or inference approaches.

Answering scientific questions sometimes requires quantifying the uncertainty of linear combinations or other functions of multiple predictions.
For example, climate scientists are interested in global average temperature, hydrologists consider the total rainfall in a catchment area, and carbon-cycle scientists want to infer CO$_2$ surface fluxes from kriged maps of atmospheric CO$_2$ concentrations. 
Joint predictive distributions at a set of prediction locations are required to quantify uncertainties of these spatial averages, totals, or other follow-up or ``downstream'' analyses. Our article details how to compute prediction variances for linear combinations of predictions under the general Vecchia approximation.
We also consider prediction of the latent process at observed locations, which is useful in spatial smoothing and in a Vecchia-Laplace approximation of generalized GPs for non-Gaussian spatial data \citep{Zilber2019}.

This article is organized as follows. Section \ref{sec:GP} reviews GP prediction. In Section \ref{sec:vecchia}, we introduce general Vecchia approximations of GP prediction. In Section \ref{sec:methods}, we discuss specific methods and study their properties. Sections \ref{sec:comparison} and \ref{sec:realdata} provide numerical comparisons using simulated and real data, respectively. We conclude in Section \ref{sec:conclusions}. Appendices \ref{app:notation}--\ref{app:proofs} contain details and proofs. A separate Supplementary Material document contains
Sections \ref{sec:compadditional}--\ref{sec:satsupp} with additional details, plots, comparisons, and a description of another Vecchia prediction method.
The proposed methods are implemented in the \texttt{R} package \texttt{GPvecchia} \citep{GPvecchia}. Code to reproduce our results will be provided with the published article.


\section{Exact Gaussian-process prediction \label{sec:GP}}

The process of interest is denoted by $\{y(\bs) \!: \bs \in \domain\}$, or $y(\cdot)$, on a continuous (i.e., non-gridded) domain $\domain \subset \mathbb{R}^d$, $d \in \mathbb{N}^+$. We assume that $y(\cdot) \sim \GP(0,K)$ is a Gaussian process (GP) with mean zero and covariance function $K: \domain \times \domain \to \mathbb{R}$, which is assumed known up to some parameters.
Let $\bs_i \in \domain$ for $i=1,\ldots,n$, and define the location vector $\locs = ( \bs_1,\ldots,\bs_n )$. For simplicity, we assume throughout that the locations in $\locs$ are unique.
Define $y_i = y(\bs_i)$ and the vectors $\by = (y_1,\ldots,y_n)$ and $\bz = (z_1,\ldots,z_n)$. The response variables $z_i$ are noisy versions of latent $y_i$: $z_i | \by \sim \normal(y_i, \tau^2_i)$ independently for all $i$.
Thus, the covariance matrix of $\by$ is $\bK = K(\locs,\locs)$, and the covariance matrix of $\bz$ is $\bC = \bK + \bD$, where $\bD$ is a diagonal matrix containing the noise or nugget variances, $\bD_{ii}=\tau_i^2$.
Define the index vector $o \subset (1,\ldots,n)$ of length $n_O=|o|$ such that the subvector $\bz_o$ contains all observed response variables (i.e., the data), and $\locs_o$ represents the vector of observed locations. (We use the vector and indexing notation described in Appendix \ref{app:notation}.) We also define $p = (1,\ldots,n) \setminus o$ to be an index vector of length $n_P = |p| = n-n_O$, such that $\locs_p$ is the vector of unobserved (prediction) locations \citep[cf.][]{Le2006}. In summary, we have:
\begin{center}
\begin{tabular}{r| l}
Notation & Terminology \\
\hline
$o \subset (1,\ldots,n)$ & vector of indices of \textbf{observed} locations \\
$p = (1,\ldots,n) \setminus o$ & vector of indices of (unobserved) \textbf{prediction} locations \\
$\by$, $\by_o$, $\by_p$ & vectors of \textbf{latent} variables \\
$\bz$, $\bz_o$, $\bz_p$ & vectors of \textbf{response} variables \\
\end{tabular}
\end{center}

Inference on unknown parameters $\bftheta$ in $K$ and $\tau_i^2$ can be carried out based on the multivariate normal likelihood,
$
\dens(\bz_o) = \normal_{n_O}(\bz_o|\bfzero,\bC_{oo}),
$
or approximations thereof.

The goal for prediction is to obtain the posterior predictive distribution of $\by$ 
via
\begin{equation}
\label{eq:predint}
\textstyle \dens(\by | \bz_o) = \int \dens(\by|\bz_o,\bftheta) \, dF(\bftheta|\bz_o).
\end{equation}
The density $\dens(\by|\bz_o,\bftheta)$ is normal with mean $\bm{\mu}(\bftheta) = \bK_{\all o}\bC_{oo}^{-1}\bz_o$ and covariance matrix
\begin{equation}
\label{eq:predcov}
\bfSigma(\bftheta) = \bK - \bK_{\all o} \bC_{oo}^{-1}\bK_{o \all},
\end{equation}
where $\bK$ and $\bC$ implicitly depend on $\bftheta$, and $\all$ denotes the vector of all indices.
When using maximum-likelihood estimation, the posterior distribution $F(\bftheta|\bz_o)$ of the parameters is effectively approximated by a point mass at $\bftheta = \hat\bftheta$ in \eqref{eq:predint}, and so $\dens(\by | \bz_o) = \normal(\by |\bfmu(\hat\bftheta),\bfSigma(\hat\bftheta))$. For Bayesian inference using MCMC, the parameter posterior in \eqref{eq:predint} is approximated as discrete uniform on, say, $\bftheta^{(1)},\ldots,\bftheta^{(L)}$, and so $\dens(\by | \bz_o) = (1/L) \sum_l \normal(\by |\bfmu(\bftheta^{(l)}),\bfSigma(\bftheta^{(l)}))$, which is often further approximated by samples from the summands.
Therefore, for both inferential paradigms, GP prediction requires obtaining $\dens(\by|\bz_o,\bftheta) = \normal(\by |\bfmu(\bftheta),\bfSigma(\bftheta))$ for particular fixed values of $\bftheta$. Here and in the following, we will thus suppress dependence on $\bftheta$ and regard it as fixed, unless stated otherwise. Sometimes (e.g., for cross validation), interest might also be in predicting $\bz_p$, but this is a trivial extension of predicting $\by_p$, in that $\bz_p | \bz_o \sim \normal(\bfmu_p,\bfSigma_{pp}+\bD_{pp})$.

Point predictions at individual locations are functions of the marginal distributions $y_i | \bz_o$ only, but quantifying the uncertainty of linear combinations (e.g., spatial averages) and generating posterior simulations require the joint posterior distribution of $\by$ given $\bz_o$.
While GP prediction is mathematically straightforward, it can be computationally expensive. The time complexity for obtaining the entire matrix $\bfSigma$ in \eqref{eq:predcov} is $\order(n_O^3 + n n_O^2 + n^2 n_O)$, and even just obtaining its diagonal elements (i.e., the prediction variances) requires $\order(n_O^3 + n n_O^2)$ time. Thus, GP prediction is computationally infeasible for large $n_O$ or $n$, and approximations or simplifying assumptions are necessary.

\section{The general Vecchia framework for GP prediction\label{sec:vecchia}}

\subsection{Definition of the framework}

The density of any random vector $\bx$ can be factored exactly as $\dens(\bx) = \allowbreak \prod_{i} \dens(x_i|x_1,\ldots,x_{i-1})$. This motivates a general Vecchia approximation \citep{Katzfuss2017a} for GP prediction, which applies Vecchia's approximation \citep{Vecchia1988} to the vector $\bx = \bz_o \cup \by$:
\begin{equation}
\label{eq:gvp}
\textstyle \adens(\bx) = \allowbreak \prod_{i=1}^{n+n_O} \dens(x_i|\bx_{g(i)}),
\end{equation}
where $g(i) \subset (1,\ldots,i-1)$ is a conditioning index vector of size $|g(i)|$, often formed based on variables with locations nearby in space to $x_i$. If $g(i) = (1,\ldots,i-1)$ for every $i$, then the exact distribution is recovered: $\adens(\bx) = \dens(\bx)$. If $|g(i)|$ is bounded by some small integer $m \ll n$, the approximation can lead to enormous computational savings, because only matrices of size $m \times m$ need to be decomposed to evaluate \eqref{eq:gvp}. Recent results \citep{Schafer2020} indicate that in some settings the approximation error can be bounded with $m$ increasing only polylogarithmically in $n$. Because the general Vecchia approximation $\adens(\bx)$ is a valid probability distribution (e.g., \citealp{Datta2016}, App.~A; \citealp{Katzfuss2017a}, Prop.~1), it can be used for approximating the posterior predictive distribution $\dens(\by|\bz_o)$ by applying the rules of probability to $\adens(\bx)$ to obtain $\adens(\by|\bz_o)$. 
If predictions at additional locations are desired later, the corresponding realizations of $y(\cdot)$ can be appended to the end of $\bx$ for consistency (see Section \ref{sec:consistent}).

The accuracy of a general Vecchia approximation depends on the choice of the ordering of the variables in $\bx$ and the specification of the conditioning index vectors $g(i)$. The computational efficiency of the approximation is governed by several factors, including the sparsity of the precision matrix for $\by|\bz_o$ and its Cholesky factor. For a given $m$, \citet{Katzfuss2017a} showed that there is often a trade-off in conditioning on latent versus response variables, in that it can be more accurate but also more computationally expensive to condition on $y_k$ rather than on $z_k$. In Section \ref{sec:methods}, we study how ordering and conditioning choices affect both the quality of the approximation and its computational burden, for the purpose of providing practical guidelines for using Vecchia's approximation for spatial prediction.

\subsection{Matrix representations \label{sec:matrix}}

In this subsection, we introduce matrix notation and recapitulate existing results \citep[e.g.,][]{Katzfuss2017a}. Let $\bQ$ be the precision matrix for $\bx$ under $\adens(\bx)$. The joint distribution implied by the approximation in \eqref{eq:gvp} is multivariate normal, $\adens(\bx) = \normal(\bfzero,\bQ^{-1})$. Define $\chol(\bM)$ to return the (lower-triangular) Cholesky factor of $\bM$, $\rev(\bM)$ to return the reverse row-column reordering of $\bM$ and $\rchol(\bM) = \rev(\chol(\rev(\bM)))$, which gives the (upper-triangular) upper-lower decomposition for $\bM$. Our notation for the relevant matrices is the following:
\begin{align*}
\begin{array}{l|l}
\mbox{Notation}  & \mbox{Note} \\
\hline
\ell = \#(\by,\bx) & \mbox{indices in } \bx \mbox{ occupied by latent variables } \by \mbox{ (i.e., } \by = \bx_\ell\mbox{)} \\
r = \#(\bz_o,\bx) & \mbox{indices in } \bx \mbox{ occupied by response variables } \bz_o \mbox{ (i.e., } \bz_o = \bx_r\mbox{)} \\
\bU = \rchol(\bQ) & \mbox{upper-lower Cholesky decomposition of } \bQ\; (\mbox{i.e., } \bQ = \bU \bU'\mbox{)} \\
\bW = \bQ_{\ell\ell} = \bU_{\ell,\all}\bU_{\ell,\all}' & \mbox{posterior precision matrix of } \by \mbox{ given } \bz_o \\
\bV = \rchol(\bW) & \mbox{upper-lower Cholesky decomposition of } \bW\; (\mbox{i.e., } \bW = \bV \bV')
\end{array}
\end{align*}
In practice, there is no need to construct the matrix $\bQ$; rather, we compute the nonzero entries of $\bU$ directly via the methods outlined in Appendix \ref{app:computeU}. From the expressions in Appendix \ref{app:computeU}, it is easy to see that $\bU$ is sparse with at most $m$ off-diagonal nonzero entries per column, and $\bU$ can be computed in $\order(nm^3)$ time.

\subsection{General Vecchia predictions \label{sec:predictions}}

The goal for GP prediction is to obtain the posterior predictive distribution of $\by$ given the response $\bz_o$, or desired summaries of this distribution.
As explained in Appendix \ref{sec:predictionsunknown}, it suffices to consider this distribution for certain values of the parameters $\bftheta$, which we again suppress for notational simplicity.
General Vecchia prediction approximates the exact conditional distribution $f(\by|\bz_o)$ as implied by the joint distribution $\adens(\bx)$ in \eqref{eq:gvp} with $\bx = \by \cup \bz_o$:
\[
\adens( \by | \bz_o ) = \frac{\adens(\bx)}{\int \adens(\bx)d\by} \equalscolon \normal_{n}(\bfmu,\bfSigma).
\]
Since $\bW = \bQ_{\ell\ell}$ is the submatrix corresponding to $\by$ of the full precision matrix $\bQ$ of $\bx = \by \cup \bz_o$, it is a well-known property of precision matrices that $\bfSigma = \bW^{-1}$. While the posterior precision matrix $\bW$ is sparse, the covariance matrix $\bfSigma$ will generally be a dense $n \times n$ matrix. Thus, it is infeasible to actually compute and store this entire matrix when $n$ is large. However, quantities of interest in the context of prediction can be computed using the general Vecchia approximation as follows:
\begin{enumerate}[itemsep=1pt,topsep=3pt]
\item The posterior mean or kriging predictor $\bfmu = \E(\by|\bz_o)$: It is straightforward to show that $\bfmu = -(\bV')^{-1}\bV^{-1}\bU_{\ell,\all} \bU_{r,\all}'\bz_o$ \citep[][proof of Prop.~2]{Katzfuss2017a}.
\item The prediction variances $\diag(\bfSigma)=(\var(y_1 | \bz_o),\ldots,\var(y_n | \bz_o))$: Based on $\bV$, a selected inversion algorithm, also referred to as the Takahashi recursions \citep{Erisman1975,Li2008,Lin2011}, can be used to compute $\bfSigma_{ij}$ for all pairs $i,j$ with $\bW_{ij} \neq 0$. Thus, it also returns the prediction variances $\bfSigma_{ii}$.
\item The joint posterior distribution of linear combinations (e.g., spatial averages): $\bH\by | \bz_o \sim \normal_k\big(\bH\bfmu,(\bV^{-1}\bH')'(\bV^{-1}\bH')\big)$, where $\bH$ is $k \times n$. As $\bV^{-1}\bH'$ is generally dense, only a moderate $k$ is computationally feasible. The variances of linear combinations can be computed faster, as $\diag(\var(\bH\by|\bz_o)) = ((\bV^{-1}\bH')\circ(\bV^{-1}\bH'))'\bfone_{n}$, where $\circ$ denotes element-wise multiplication and $\bfone_n$ is an $n$-vector of ones.
\item Conditional simulation from the posterior predictive distribution $\normal(\bfmu,\bfSigma)$: Draw $n$ i.i.d.\ samples $a_i \sim \normal(0,1)$ from the standard normal distribution, set  $\ba= (a_1,\ldots,a_{n})'$ and $\by^* = \bfmu + (\bV')^{-1}\ba$. Then $\by^* \sim \normal(\bfmu,\bfSigma)$.
\end{enumerate}
All of these tasks require computation of $\bV = \rchol(\bW)$ from $\bU$. The cost of this Cholesky factorization depends on the number of nonzero entries per column in $\bV$. In general, it is crucial for fast predictions for large $n$ that $\bV$ is sufficiently sparse. Computational complexity is discussed in Section \ref{sec:complexity} in more detail.

\section{Specific methods and their properties\label{sec:methods}}

We now consider two different ordering schemes, response-first and latent-first ordering, within the general Vecchia framework for fast GP prediction, along with several special cases. The methods are summarized in Table \ref{tab:methods} and illustrated in Figure \ref{fig:toyillustration}. In Section \ref{sec:sgvp}, we consider an additional approach based on a third ordering scheme, which is an extension of the sparse general Vecchia likelihood approximation of \citet{Katzfuss2017a}, but this approach is less suitable for GP prediction. The general Vecchia framework allows a coherent, systematic discussion and comparison of the different methods, which has so far been lacking in the literature on scalable GP predictions.

In our framework, the vector $\locs$ is obtained based on an ordering of the unordered set of locations $\{ \bs_1,\ldots,\bs_n \}$. For most of the methods below, we assume an observed-prediction (OP) restriction, meaning that the observed locations are ordered first and prediction locations last; that is, $o=(1,\ldots,n_O)$ and $p=(n_O+1,\ldots,n)$. Unless stated otherwise, we recommend and use a maximum-minimum distance (maxmin) ordering \citep{Guinness2016a,Schafer2017} constrained to order the prediction locations last. The latent and response variables are then ordered according to how the locations are ordered, with $y_i$ and $z_i$ corresponding to $\bs_i$. After ordering $\by$ and $\bz$ according to the locations, we must consider the separate issue of how to order $\by$ and $\bz_o$ within $\bx$. We study the impact that this choice has on approximation accuracy and computational cost.


This section also studies the impact of how the conditioning sets are chosen. The Vecchia approximation assumes that each variable $x_i$ only conditions on variables that are ordered previously in $\bx$. Of those previously ordered variables, we recommend considering those corresponding to the nearest $m$ locations, potentially under additional restrictions. However, if $\bs_j$ is one of the nearest $m$ locations, and both $y_j$ and $z_j$ are ordered before $x_i$, then we only consider one of them, because of conditional independence of $x_i$ and $z_j$ given $y_j$. 
Similarly, if $y_i$ is ordered before $z_i$, then $z_i$ will only condition on $y_i$, because $z_i$ is conditionally independent of all other variables in $\bx$ given $y_i$.


\begin{table}
\centering
\begin{tabular}{l|l|l|l|l|l}
Method & related to & reference & cnvrg. & linear & recommended \\
\hline
RF-full   &                      &                         & \cmark & \cmark & 2D, large $n_O, n_P$\\
RF-stand  & standard Vecchia     & \citet{Guinness2016a}   & \cmark & \cmark & 2D, $n_P \ll n_O$ \\
RF-ind    & NNGPR, local kriging & \citet{Finley2017}      & \xmark & \cmark &        \\
LF-full   & NNGP with ref.\ set $\locs$ & \citet{Datta2016}& \cmark & \xmark & 2D, small $n_O$ \\
LF-ind    & NNGPC                & \citet{Finley2017}      & \xmark & \xmark &       \\
LF-auto   &                      &                         & \cmark & \cmark & 1D    \\  
\end{tabular}
\caption{Summary of considered methods, with details on response-first ordering (RF) given in Section \ref{sec:responsefirstordering} and on latent-first ordering (LF) in Section \ref{sec:latentfirstordering}. cnvrg.: convergence to exact joint distribution, $\adens(\by|\bz_o) \rightarrow \dens(\by|\bz_o)$ as conditioning-set size $m \rightarrow n$; linear: computational complexity guaranteed to be linear in $n$ for fixed $m$; NNGP: nearest-neighbor Gaussian process; NNGPR: NNGP-response; NNGPC: NNGP-collapsed.}
\label{tab:methods}
\end{table}

\begin{figure}
\hspace*{\fill}
 \begin{subfigure}{.13\textwidth}
  \small RF-full
 \end{subfigure}
 \begin{subfigure}{.6\textwidth}
  \includegraphics[trim={62mm 15mm 62mm 15mm},clip,width =1\linewidth]{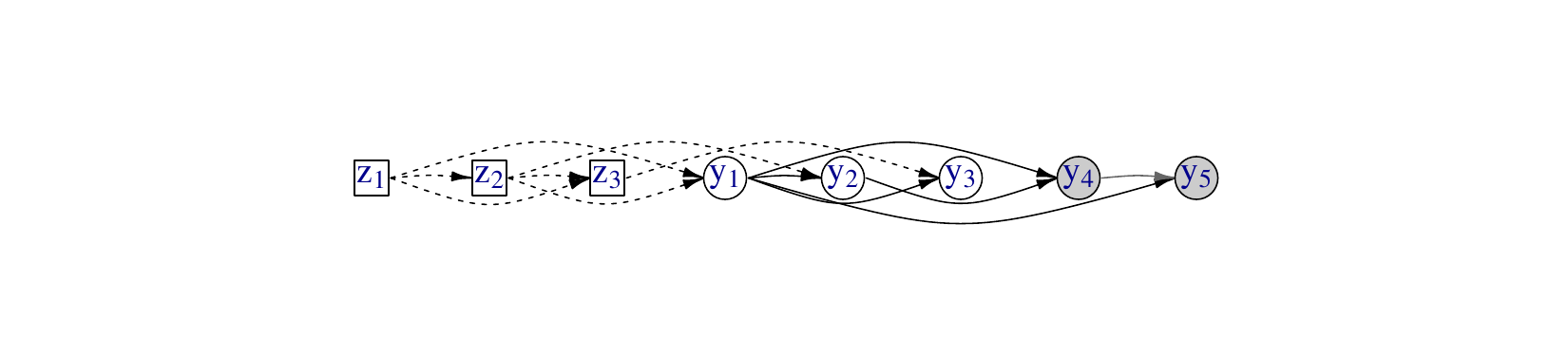} 
 \end{subfigure}
\hspace*{\fill}
\\
\hspace*{\fill}
 \begin{subfigure}{.13\textwidth}
  \small RF-stand
 \end{subfigure}
 \begin{subfigure}{.6\textwidth}
  \includegraphics[trim={62mm 15mm 62mm 15mm},clip,width =1\linewidth]{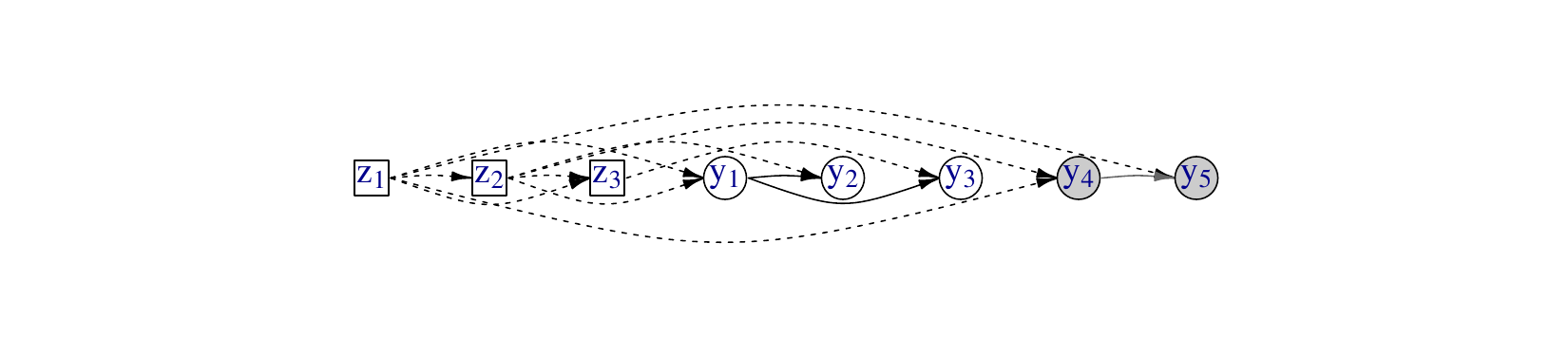} 
 \end{subfigure}
\hspace*{\fill}
\\
\hspace*{\fill}
 \begin{subfigure}{.13\textwidth}
  \small RF-ind
 \end{subfigure}
 \begin{subfigure}{.6\textwidth}
  \includegraphics[trim={62mm 15mm 62mm 15mm},clip,width =1\linewidth]{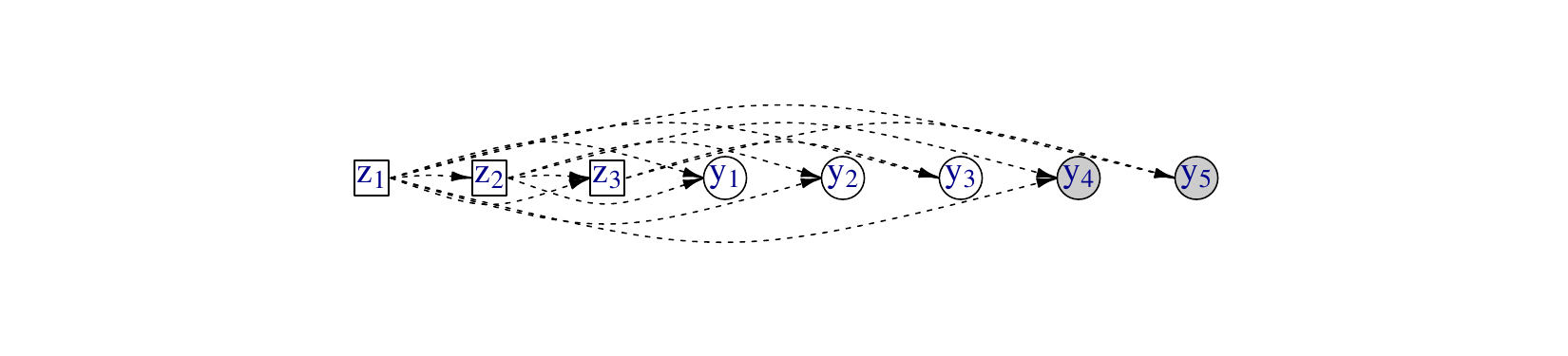} 
 \end{subfigure}
\hspace*{\fill}
\\
\hspace*{\fill}
 \begin{subfigure}{.13\textwidth}
  \small LF-full
 \end{subfigure}
 \begin{subfigure}{.6\textwidth}
  \includegraphics[trim={62mm 15mm 62mm 15mm},clip,width =1\linewidth]{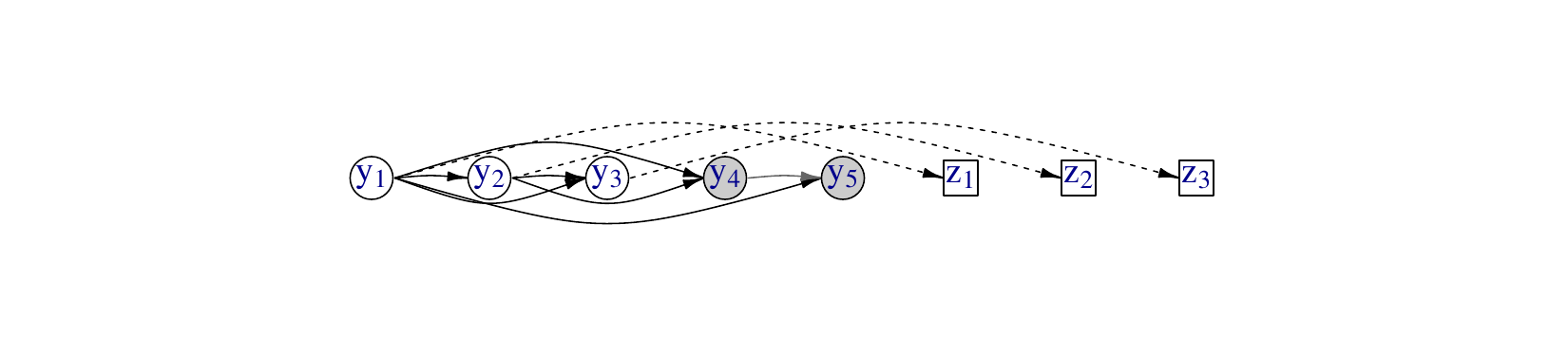} 
 \end{subfigure}
\hspace*{\fill}
\\
\hspace*{\fill}
 \begin{subfigure}{.13\textwidth}
  \small LF-ind
 \end{subfigure}
 \begin{subfigure}{.6\textwidth}
  \includegraphics[trim={62mm 15mm 62mm 15mm},clip,width =1\linewidth]{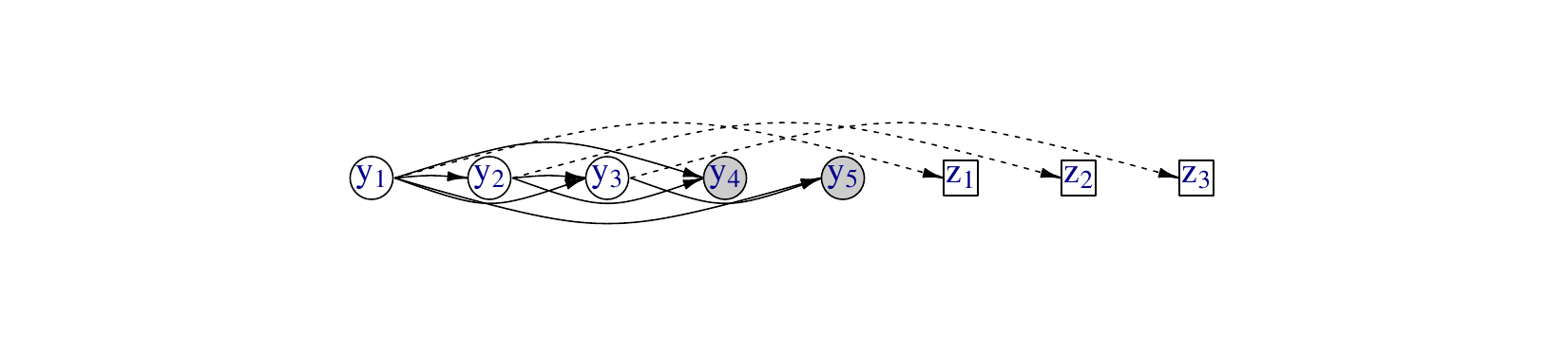} 
 \end{subfigure}
\hspace*{\fill}
\\
\hspace*{\fill}
 \begin{subfigure}{.13\textwidth}
  \small LF-auto
 \end{subfigure}
 \begin{subfigure}{.6\textwidth}
  \includegraphics[trim={62mm 15mm 62mm 15mm},clip,width =1\linewidth]{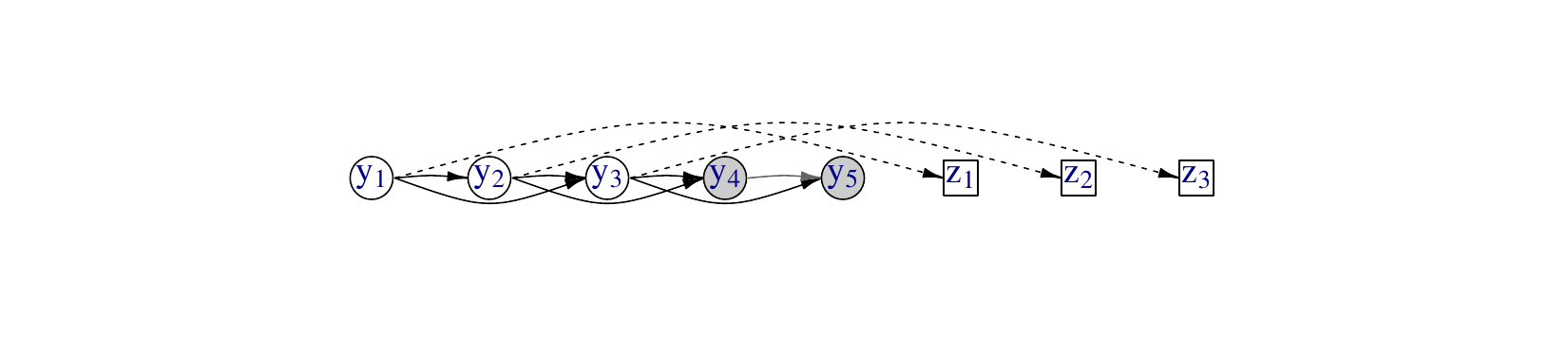} 
 \end{subfigure}
\hspace*{\fill}
 \caption{Illustration of the methods in Table \ref{tab:methods} as directed acyclic graphs (DAGs) for a toy example with $n=5$, $n_O=3$, $n_P=2$, and conditioning sets of size $m=2$. Variable ordering is from left to right, and conditioning indicated by arrows. Response variables $z_i$ with $i \in o$ are in squares, with dashed arrows to and from response variables. Prediction variables $y_i$ with $i \in p$ are in grey, as are arrows between them.}
 \label{fig:toyillustration}
 \end{figure}

\subsection{Response-first ordering\label{sec:responsefirstordering}}

Response-first ordering means that $\bx$ is ordered as $\bx = (\bx_r,\bx_\ell) = (\bz_o, \by)$, allowing us to rewrite \eqref{eq:gvp} as
\begin{equation}
    \label{eq:rf}
   \textstyle \adens(\bx)  = \left(\prod_{i=1}^{n} \dens(y_i | \by_{q_y(i)}, \bz_{q_z(i)} )  \right) \left(\prod_{i \in o} \dens(z_i | \bz_{g(i)} )\right),
\end{equation}
where $q_y(i)$ and $q_z(i)$ are index sets implied by $g(j)$ for $y_i = \bx_j$; for example, if $k \in g(j)$ and $\bx_k = z_l$, then $l \in q_z(i)$.
Under response-first ordering, $\bV = \bU_{\ell\ell}$ is simply a submatrix of $\bU$.
To see this, note that $\bU$ can be written in block form as
\begin{equation}
\label{eq:Urf}
\bU = \begin{bmatrix}
\bU_{rr} & \bU_{r\ell} \\
\bU_{\ell r} & \bU_{\ell\ell}
\end{bmatrix}
= \begin{bmatrix}
\bU_{rr} & \bU_{r\ell} \\
\bfzero & \bU_{\ell\ell}
\end{bmatrix},
\end{equation}
where $\bU$ and hence $\bU_{\ell\ell}$ are upper-triangular. Therefore, $\bW = \bU_{\ell,\all}\bU_{\ell,\all}' = \bU_{\ell\ell}\bU_{\ell\ell}'$, and so $\bV = \rchol(\bW) = \bU_{\ell\ell}$.
Thus, after constructing $\bU$, no additional computation is required to obtain $\bV$, so predictions can be computed in linear time. It is important to use this result to fill the entries of $\bV$ directly, rather than forming $\bW$ and factoring it. The latter approach can lead to a large number of numerical nonzeros in $\bV$, which are symbolic nonzero entries in $\bV$ that are zero in theory but nonzero in practice due to numerical errors.
These numerical nonzeros are illustrated in Figure \ref{fig:V2D}, and described in detail in Appendix \ref{app:nonzeros}.

For the following three methods, we consider response-first ordering under OP restriction (see beginning of Section \ref{sec:methods}), meaning that $\by = (\by_o,\by_p)$, and so $\bx = (\bz_o,\by_o,\by_p)$.

\paragraph{Response-first ordering, full conditioning (RF-full)} 

This scheme is labeled as ``full'' because we allow every variable to condition on any variables ordered previously in $\bx$. 
In \eqref{eq:rf}, the conditioning vectors $\by_{q_y(i)}$ and $\bz_{q_z(i)}$ are chosen as the $m$ variables closest in space to $y_i$, among those that are previously ordered in $\bx$, conditioning on the latent $y_j$ instead of the response $z_j$ whenever possible. Specifically, we set $q(i)$ to consist of the indices corresponding to the $m$ nearest locations to $\bs_i$, including $i$ for $i \in o$, and not including $i$ for $i \in p$. Then, for any $j \in q(i)$, we let $y_i$ condition on $y_j$ if it is ordered previously in $\bx$, and condition on $z_j$ otherwise. More precisely, we set $q_y(i) = \{j \in q(i): j<i\}$ and $q_z(i) = \{j \in q(i): j\geq i\}$.

\paragraph{Response-first ordering, standard conditioning (RF-stand)}

This scheme is identical to RF-full except that $\by_p$ conditions only on $\by_p$ and $\bz_o$, not on $\by_o$. More precisely, we use the same $q(i)$ as in RF-full, but then set $q_y(i)= \{j \in q(i): j<i; \, j \in p\}$ and $q_z(i) = q(i) \setminus q_y(i)$. This approach is labeled as ``standard'' because the posterior mean and conditional simulations of $\by_p$ from this model are equivalent to those obtained in \cite{Guinness2016a}, which used the standard Vecchia approximation. 

RF-stand is computationally useful if only predictions of $\by_p$ (not of $\by_o$) are desired, because $\by_o$ can be removed from $\bx$ without changing the predictions of $\by_p$.
Specifically, we then have $\bU_{\ell_o \ell_p}=\bfzero$ from \eqref{eq:U}. It can be shown that $\bW = \blockdiag(\bW_{oo},\bW_{pp})$ and $\bV = \blockdiag(\bV_{oo},\bV_{pp})$ are both block-diagonal with $\bV_{pp} = \bU_{\ell_p \ell_p}$, and $\bfmu_p = (\bU_{\ell_p \ell_p}')^{-1}(\bU_{\ell_o \ell_p})'\bz_o$. 
Thus, when only prediction at unobserved locations $\locs_p$ is desired, the prediction tasks laid out in Section \ref{sec:predictions} can be carried out solely based on $\bU_{\all \ell_p}$, which is the submatrix formed by the last $n_P$ columns of $\bU$ corresponding to $\by_p$. 
That is, the first $2n_O$ columns of $\bU$ corresponding to $\by_o$ and $\bz_o$ would then not be required for prediction, resulting in a prediction complexity that depends on $n_P$, not on $n=n_O+n_P$ (once the $q(i)$ have been determined).
This computational simplification comes at the price of some loss of accuracy relative to RF-full because of the restriction of the conditioning sets (see Proposition \ref{prop:KLordering} and Section \ref{sec:addcomp}).

\paragraph{Response-first ordering, independent conditioning (RF-ind)}

RF-ind also uses the ordering $(\bz_o,\by_o,\by_p)$, but we enforce that latent variables condition only on $\bz_o$, and so $q_y(i) = \emptyset$ for all $i=1,\ldots,n$ in \eqref{eq:rf}. Then, $q_z(i)$ consists of the indices corresponding to the $m$ nearest observed locations to $\bs_i$ among $\locs_o$, including $\bs_i$ if $i\in o$. RF-ind ignores any posterior dependence between entries of $\by$. More precisely, note that, from \eqref{eq:U}, $\bU_{\ell\ell}$ and hence $\bW$ are now diagonal, and so:
\begin{align*}
\bfSigma & = \bW^{-1} = \diag(\{\bU_{\ell_i \ell_i}^2: i=1,\ldots,n \})^{-1} 
 = \diag\big(\{\var(\by_{i}|\bz_{q_z(i)}): i=1,\ldots,n \}\big).
\end{align*}
Similarly, it is straightforward to show that $\bfmu_{i} = \E(y_{i} | \bz_{q_z(i)})$. RF-ind is equivalent to local kriging, in which each marginal predictive distribution is obtained by considering the conditional distribution of $y_i$ given the neighboring observations from $\bz_o$. This is implicitly the same predictions as in the NNGP-response model in \cite{Finley2017}, except that we here predict $\by_p$ instead of predicting $\bz_p$ as in the NNGPR. The latter can be easily achieved by adding the noise or nugget variance to each prediction variance (see end of Section \ref{sec:GP}). RF-ind can in principle be extremely fast in parallel computing environments because each conditional mean and variance can be calculated completely in parallel.

\subsection{Latent-first ordering\label{sec:latentfirstordering}}

Latent-first ordering means that $\bx$ is ordered as $\bx = (\by,\bz_o)$, resulting in the approximation
\begin{equation}
    \label{eq:lf}
\textstyle \adens(\bx)  = \left(\prod_{i=1}^{n} \dens(y_i | \by_{q_y(i)} )  \right) \left(\prod_{i \in o} \dens(z_i | y_i )\right),
\end{equation}
where $z_i$ conditions only on $y_i$ because of conditional independence from all other variables in the model.

Under latent-first ordering, $\bV$ cannot be obtained directly as a submatrix of $\bU$ as in response-first ordering. However, under the OP restriction that $\by$ is ordered as $\by = (\by_o, \by_p)$, $\bU$ does contain some exploitable structure. Note that under this ordering, we have $\ell_o=o$, $\ell_p=p$, and $\bU_{pr}=\bfzero$. Hence,
\begin{equation}
\label{eq:Vlf}
\bU = \begin{bmatrix}
\bU_{oo} & \bU_{op} & \bU_{or} \\
\bm{0}   & \bU_{pp} & \bm{0} \\
\bm{0}   & \bm{0}   & \bU_{rr}
\end{bmatrix}, \quad
\bW = \begin{bmatrix}
\bW_{oo}  & \bU_{op}\bU_{pp}' \\
\bU_{pp}\bU_{op}' & \bU_{pp}\bU_{pp}'
\end{bmatrix}, \quad
\bV = \begin{bmatrix}
\bV_{oo}  & \bU_{op} \\
\bfzero & \bU_{pp}
\end{bmatrix},
\end{equation}
where only the entries of $\bV_{oo}=\rchol(\bW_{oo})$ must be computed from  $\bW_{oo}=\bU_{oo}\bU_{oo}'+\bU_{or}\bU_{or}'$, and the last $n_P$ columns of $\bV$ corresponding to $\by_p$ can simply be ``copied'' from $\bU$.
This result can make latent-first orderings computationally competitive when there are many more prediction locations than observation locations, for example when predictions are required on a fine grid from a small number of observations.

\paragraph{Latent-first ordering, full conditioning (LF-full)}

We consider latent-first ordering under OP restriction: $\bx = (\by_o,\by_p,\bz_o)$. As in RF-full, LF-full scheme is labeled as ``full'' because we allow every variable to condition on any variables ordered previously in $\bx$. Specifically, in \eqref{eq:lf}, the latent conditioning vector $\by_{q_y(i)}$ simply consists of the $m$ latent variables $y_j$ with $j<i$ whose locations are closest in space to $\bs_i$. 

We have found that LF-full is usually the most accurate scheme in the examples we have tried; however, it can also be the most computationally demanding. As shown in \eqref{eq:Vlf}, only parts of $\bV$ can be recovered directly from $\bU$, and  linear computational complexity cannot be guaranteed in general; see Section \ref{sec:complexity} for more details. 

LF-full can be viewed as a special version of the NNGP in \citet{Datta2016}, in which their reference set is chosen as $\locs$, the vector of observed and prediction locations.

\paragraph{Latent-first ordering, independent conditioning (LF-ind)}

This scheme is a special case of LF-full with $q_y(i) \subset o$, and hence there is no conditioning on variables at prediction locations: $q_y(i) \cap p = \emptyset$. This assumption of conditional independence of the entries of $\by_p$ given $\by_o$ can lead to inaccurate approximations of the joint predictive distribution at a set of locations (see, e.g., top row in Figure \ref{fig:KL2D}).
As with LF-full, linear computational complexity cannot be guaranteed because the submatrix of $\bV$ corresponding to $\by_o$ cannot be obtained directly from $\bU$. LF-ind is implicitly the same approximation as used in the NNGP-collapsed model in \cite{Finley2017}.

\paragraph{Multi-resolution approximation (MRA)}

Predictions using the MRA \citep{Katzfuss2015,Katzfuss2017b} can be viewed as a version of LF-full, except that the $q_y(i)$ are based on iterative domain partitioning \citep[cf.][Sect.~3.7]{Katzfuss2017a}. In contrast to the nearest-neighbor conditioning in LF-full, MRA conditioning ensures sparsity and linear complexity, which can be shown using \citet[][Prop.~6]{Katzfuss2017a}. While RF-full can be more accurate than MRA for a given $m$ (Section \ref{sec:addcomp}), the special conditioning structure of the MRA has many other benefits \citep{Jurek2018}. For example, $\bV^{-1}$ has the same sparsity as $\bV$ for the MRA, and so $\bV^{-1}\bH'$ is generally sparse if $\bH$ is sparse. This allows computing the posterior covariance matrix of a large number of linear combinations $\bH\by$.

\paragraph{Latent-first and coordinate ordering, autoregressive conditioning (LF-auto)}

Finally, we consider a method that is based on a general ordering of the locations, without OP restriction. We consider latent-first ordering, $\bx=(\by,\bz_o)$, and the approximation in \eqref{eq:lf}, where each $y_i$ simply conditions on $(y_{i-m},\ldots,y_{i-1})$; more precisely, we set $q_y(i) = \{\max(1,i-m),\ldots,i-1\}$. This amounts to a latent autoregressive process of order $m$.
It is easy to verify that $\bW$ is banded with bandwidth $m$, and so its Cholesky factor $\bV$ can be obtained in $\order(nm^2)$ time.

LF-auto is most appropriate when successive locations in $\locs$ (and hence variables in $\bx$) tend to be close in space, so that $\by_{q_y(i)}$ has strong correlation with $y_i$. This is often the case with coordinate-based ordering, especially in one-dimensional space. Therefore, we only consider LF-auto based on left-to-right ordering in one-dimensional domains ---  see Figure \ref{fig:1Dillus} for an illustration.
GPs are often used to model functions in one dimension, for example in astronomy \citep[e.g.,][]{Wang2012,Kelly2014,Foreman-Mackey2017}.

\begin{figure}
	\begin{subfigure}{1\textwidth}
	\centering
	\centering
	\includegraphics[width =.67\linewidth]{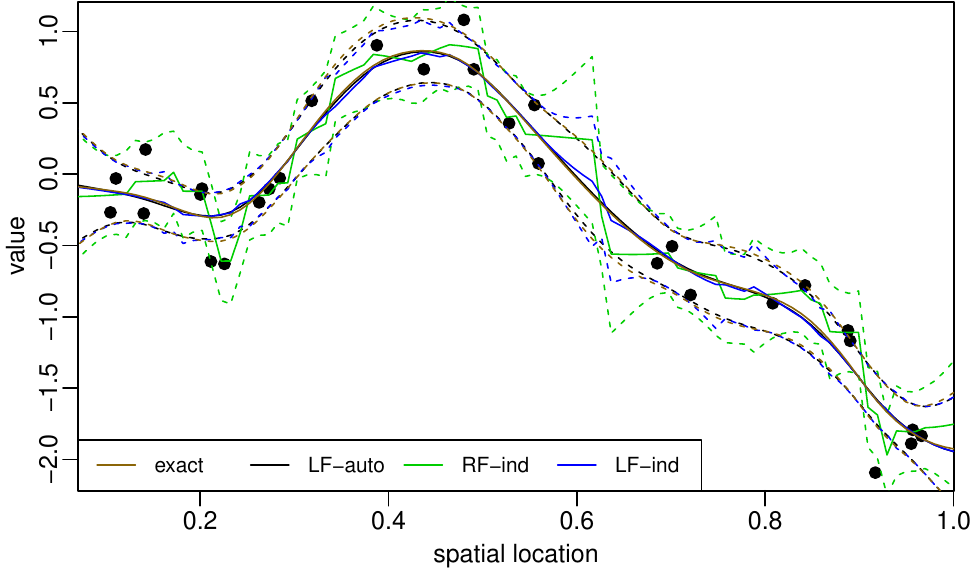}
  \caption{Simulated data (dots), posterior means (solid lines) and pointwise 95\% intervals (dashed lines)}
\label{fig:predillus}
	\end{subfigure} \\ \vspace{4mm}

\begin{subfigure}{.2\textwidth}
	\centering
	\includegraphics[width =.98\linewidth]{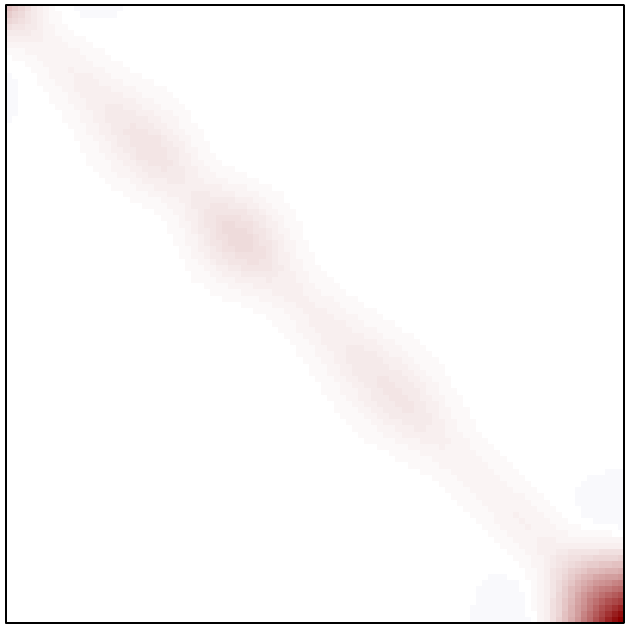}
	\caption{exact}
	\label{fig:cov1Dexact}
	\end{subfigure}%
\hfill
	\begin{subfigure}{.2\textwidth}
	\centering
	\includegraphics[width =.98\linewidth]{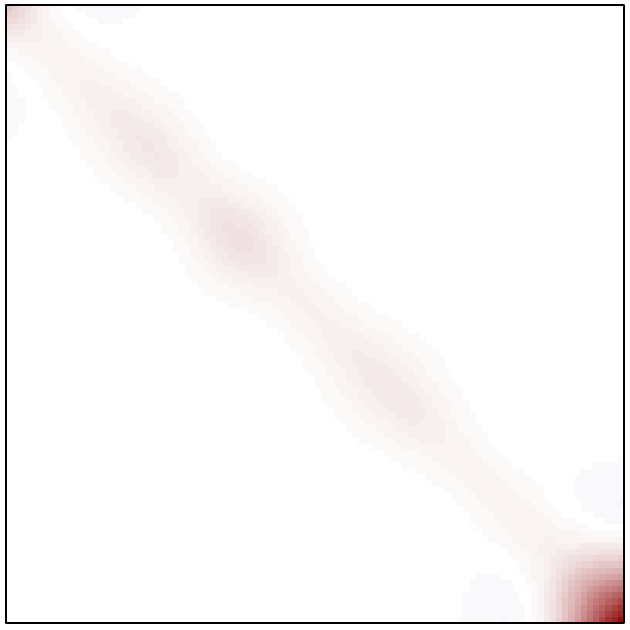}
	\caption{LF-auto}
	\end{subfigure}%
\hfill
	\begin{subfigure}{.2\textwidth}
	\centering
	\includegraphics[width =.98\linewidth]{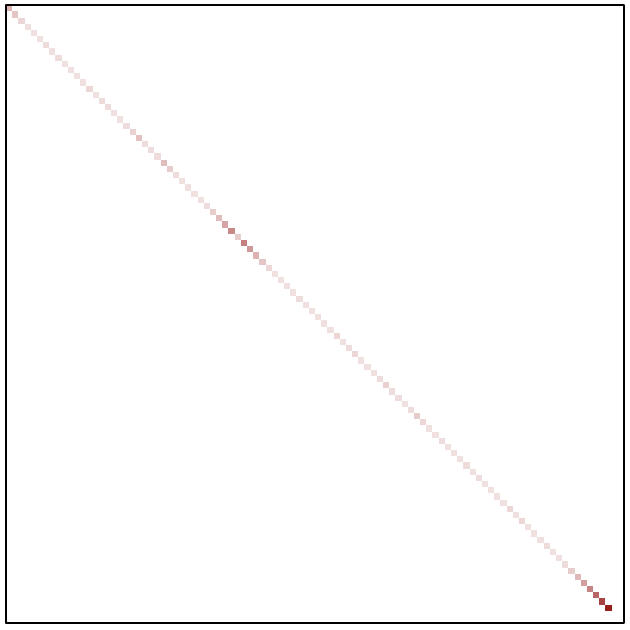}
	\caption{RF-ind}
	\end{subfigure}%
\hfill
	\begin{subfigure}{.2\textwidth}
	\centering
	\includegraphics[width =.98\linewidth]{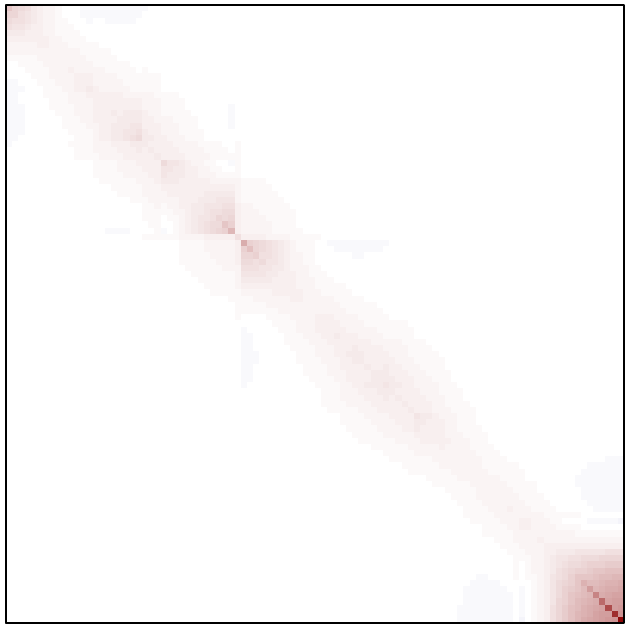}
	\caption{LF-ind}
	\label{fig:cov1DLF-ind}
	\end{subfigure}%
\hfill
	\begin{subfigure}{.081\textwidth}
	\centering
	\includegraphics[width =.8\linewidth]{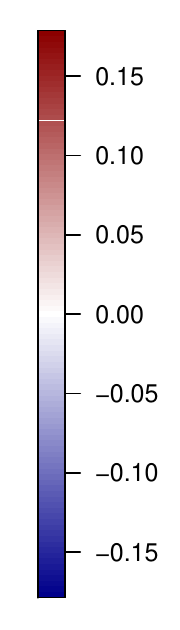}
	\end{subfigure}%
  \caption{Illustration of predictions at $n_P=100$ locations based on $n_O=30$ simulated data on $\domain = [0,1]$ for a GP with Mat\'ern covariance with smoothness 1.5 using $m=2$ and coordinate (left-to-right) ordering. 
  (\subref{fig:predillus}): LF-auto lines are largely covered up by exact lines.
  (\subref{fig:cov1Dexact})--(\subref{fig:cov1DLF-ind}): Posterior predictive covariance matrices $\bfSigma_{pp}$}
\label{fig:1Dillus}
\end{figure}

\subsection{Summary and properties\label{sec:properties}}

We now discuss some of the properties of the methods summarized in Table \ref{tab:methods}.

\subsubsection{Zero noise}

In the case of zero noise or nugget (i.e., $\tau_i =0$ for all $i=1,\ldots,n$), we have $\by_o = \bz_o$, and so there is no real distinction between RF-full, RF-stand, and LF-full anymore. Similarly, LF-ind and RF-ind then become equivalent.

\subsubsection{Approximation accuracy\label{sec:approxacc}}

Each of our methods produces a prediction $\adens(\by|\bz_o)$, which can be viewed as an approximation of the exact GP prediction $\dens(\by|\bz_o)$, or as a valid and exact conditional distribution implied by the multivariate normal Vecchia density $\adens(\bx)$ in \eqref{eq:gvp}. We now discuss the Vecchia-approximation error from the former perspective, in terms of the (conditional) Kullback-Leibler (KL) divergence to the exact distribution $\dens$. The KL divergence is the expected difference in log-likelihood:
\[
\textstyle \KL\big(\dens(\bx)\|\adens(\bx)\big) = \E^\bx \log \dens(\bx) - \E^\bx \log\adens(\bx) = \int \dens(\bx) \log\big(\dens(\bx)/\adens(\bx)\big) d\bx.
\]
Similarly, for generic random vectors $\bz$ and $\by$, we define the conditional KL (CKL) divergence as
\[
\textstyle \CKL\big(\dens(\by|\bz)\|\adens(\by|\bz)\big) = \E^\bz \KL\big(\dens(\by|\bz)\|\adens(\by|\bz)\big) = \int \dens(\bz) \int \dens(\by|\bz) \log\big(\dens(\by|\bz)/\adens(\by|\bz)\big) d\by d\bz,
\]
which is the KL divergence for the conditional distribution of $\by$ given $\bz$, averaged over possible realizations of $\bz$.

\begin{prop}
\label{prop:KLvecchia}
Let $\bx=(\bx^{(1)},\bx^{(2)})$ be a multivariate normal random vector, with $\bx^{(1)} = (x_1,\ldots,x_{n_1})$ and $\bx^{(2)} = (x_{n_1+1},\ldots,x_{n_1+n_2})$. Further, let $\adens_1(\bx)$ and $\adens_2(\bx)$ be two Vecchia approximations of the form \eqref{eq:gvp}, with conditioning vectors $g_1(i)$ and $g_2(i)$, respectively.
\begin{enumerate}
    \item If $g_1(i) \subset g_2(i)$ for all $i=1,\ldots,n_1+n_2$, then $\KL\big(\dens(\bx)\|\adens_1(\bx)\big) \geq \KL\big(\dens(\bx)\|\adens_2(\bx)\big)$.
    \item If $g_1(i) \subset g_2(i)$ for all $i=n_1+1,\ldots,n_1+n_2$, then $\CKL\big(\dens(\bx^{(2)}|\bx^{(1)})\|\adens_1(\bx^{(2)}|\bx^{(1)})\big) \geq \CKL\big(\dens(\bx^{(2)}|\bx^{(1)})\|\adens_2(\bx^{(2)}|\bx^{(1)})\big)$.
\end{enumerate}
\end{prop}
Using these properties, the following can be said about the approximation accuracy of the general Vecchia prediction methods in Table \ref{tab:methods}:
\begin{prop}
\label{prop:KLordering}
Let $\KL_m^{\text{A}}(\bx) = \KL(\dens(\bx)\|\adens(\bx))$, where $\adens(\bx)$ is the approximate density obtained using Vecchia approach A (e.g., A=RF-full) with conditioning vectors of size $m$, and similarly for $\CKL_m^{\text{A}}$.
\begin{enumerate}
    \item $\KL^{\text{A}}_{n-1}(\bx)=\CKL^{\text{A}}_{n-1}(\by|\bz_o)=\CKL^{\text{A}}_{n-1}(\by_p|\bz_o) = 0$ for $\text{A} \in \{\text{RF-full,RF-stand,LF-full,LF-auto}\}$
    \item $\KL^{\text{A}}_{m+1}(\bx)\leq\KL^{\text{A}}_{m}(\bx)$, for all A in Table \ref{tab:methods}
    \item $\CKL^{\text{A}}_{m+1}(\by_p|\by_o,\bz_o)\leq\CKL^{\text{A}}_{m}(\by_p|\by_o,\bz_o)$ for $\text{A} \in \{\text{RF-full,RF-stand,RF-ind,LF-full,LF-ind}\}$
    \item $\CKL^{\text{A}}_{m+1}(\by|\bz_o)\leq\CKL^{\text{A}}_{m}(\by|\bz_o)$ for $\text{A} \in \{\text{RF-full,RF-stand,RF-ind}\}$
    \item $\CKL^{\text{A}}_{m+1}(\by_p|\bz_o)\leq\CKL^{\text{A}}_{m}(\by_p|\bz_o)$ for $\text{A} \in \{\text{RF-stand,RF-ind}\}$
    \item $\KL^{\text{RF-full}}_{m}(\bx) \leq \KL^{\text{RF-stand}}_{m}(\bx)$ and $\CKL^{\text{RF-full}}_{m}(\by|\bz_o) \leq \CKL^{\text{RF-stand}}_{m}(\by|\bz_o)$
    \item For a GP with exponential covariance function in one dimension: $\KL^{\text{LF-auto}}_{1}(\bx)=
    0$
\end{enumerate}
\end{prop}
To summarize, the Vecchia approximations tend to become more accurate as the conditioning-set size $m$ increases. For all approaches in Table \ref{tab:methods} this increase in accuracy is in terms of the KL divergence for $\dens(\bx)$, but for certain methods the same is also guaranteed for distributions such as $\dens(\by|\bz_o)$ or $\dens(\by_p|\bz_o)$ that are more relevant for prediction. 
Indeed, we have $\adens(\by|\bz_o) = \dens(\by|\bz_o)$ in the limit of $m = n-1$ for most methods. However, for RF-ind and LF-ind, due to the assumption of conditional independence of the entries in $\by_p$ given $\bz_o$ and $\by_o$, this holds only marginally, in the sense that $\adens(\by_i|\bz_o) = \dens(\by_i|\bz_o)$ but $\adens(\by|\bz_o) \neq \dens(\by|\bz_o)$. This is also shown numerically in Figures \ref{fig:KL1D} and \ref{fig:KL2D}, top row.
RF-full can be expected to result in more accurate approximations of $\dens(\by|\bz_o)$ than RF-stand.
LF-auto is exact with $m\geq 1$ for a GP with exponential covariance function in one dimension. More generally, for a GP with Mat\'ern covariance with smoothness $\nu$, we obtain nearly exact representations if $m>\nu$ \citep[cf.][Fig.~2a]{Katzfuss2017a}.

\subsubsection{Computational complexity \label{sec:complexity}}

As laid out in Section \ref{sec:predictions}, the relevant quantities for prediction can be obtained by computing $\bU$, calculating $\bV$ from $\bU$, carrying out a selected inversion based on $\bV$, and performing triangular solves in $\bV$.
For all methods, the matrix $\bU$ has only $m$ off-diagonal nonzero elements per column by construction, and it can be computed in $\order(nm^3)$ time using \eqref{eq:U}. The cost for each triangular solve in $\bV$ is on the order of the number of nonzero entries in $\bV$. The cost of computing $\bV$ and the selected inversion is proportional to the sum of squares of the number of nonzero entries per column in $\bV$.

For all response-first methods (i.e., RF-full, RF-stand, and RF-ind), we have shown below \eqref{eq:Urf} that $\bV$ is simply a submatrix of $\bU$. For LF-auto and MRA, $\bV =\rchol(\bW)$ must be explicitly computed, but these methods' special conditioning structures ensure that there is no fill-in. Hence, for these methods the number of off-diagonal nonzero entries in each column of $\bV$ is guaranteed to be at most $m$. This implies that $\bV$ and the selected inverse can hence be computed in $\order(n m^2)$ time, $\bV^{-1}\bH'$ in $\order(n m \ell)$ time, and $(\bV^{-1}\bH')'(\bV^{-1}\bH')$ can be computed in $\order(n\ell^2)$ time. 
(An additional approximation can be necessary for selected inversion in the case of OP methods --- see Appendix \ref{app:selinv} for details.)
Thus, the complexity of GP prediction using LF-auto, MRA, and all response-first methods is linear in $n$.

In contrast, LF-full and LF-ind result in more nonzero entries in $\bW$, and hence more potential for fill-in in $\bV$ (see Figures \ref{fig:Vsparsity} and \ref{fig:V2D} for illustration). As a result, linear complexity is not guaranteed for these two methods, even when using the block form in \eqref{eq:Vlf}. For example, for locations on a regular two-dimensional grid ordered according to their coordinates, \citet[][Prop.~5]{Katzfuss2017a} proved that the time required for computing $\bV_{oo} = \rchol(\bW_{oo})$ using reverse ordering grows quadratically with $n_O$. As shown in Figure \ref{fig:timing}, while fill-reducing permutations can reduce computation times somewhat, inference might still be slow or fail due to memory issues.

For the latent NNGP model underlying LF-ind and LF-full, \citet{Datta2016} proposed a sequential Gibbs sampler that samples every element of $\by$ from its full-conditional distribution. As this approach can lead to non-convergence issues, \citet{Finley2017} instead proposed to sample $\by_o$ jointly, and then to sample $\by_p$ conditional on $\by_o$. This avoids numerical nonzero entries similarly to our block-form $\bV$ in \eqref{eq:Vlf}, but it still requires the expensive factorization of $\bW$, and it can result in additional sampling error.

\begin{figure}
	\centering
	\begin{subfigure}{.48\textwidth}
	\centering
	\includegraphics[width =.97\linewidth]{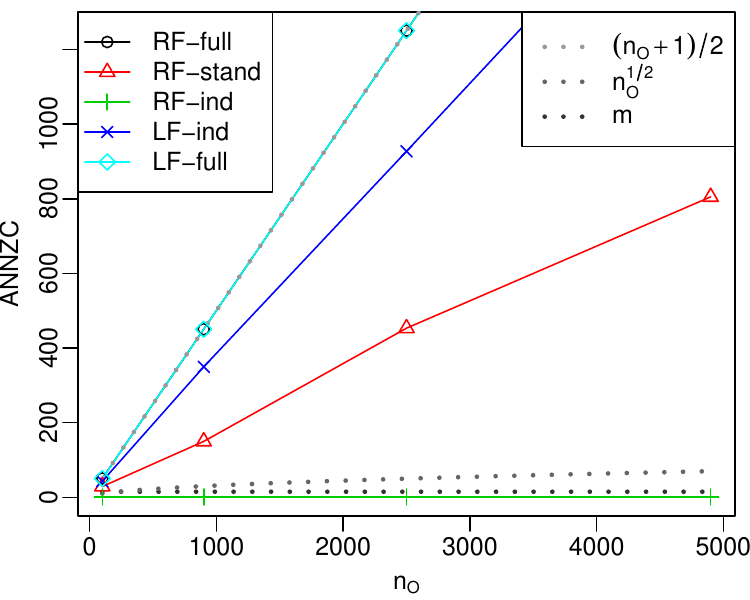}
	\caption{standard Cholesky}
	\end{subfigure}%
\hfill
	\begin{subfigure}{.48\textwidth}
	\centering
	\includegraphics[width =.97\linewidth]{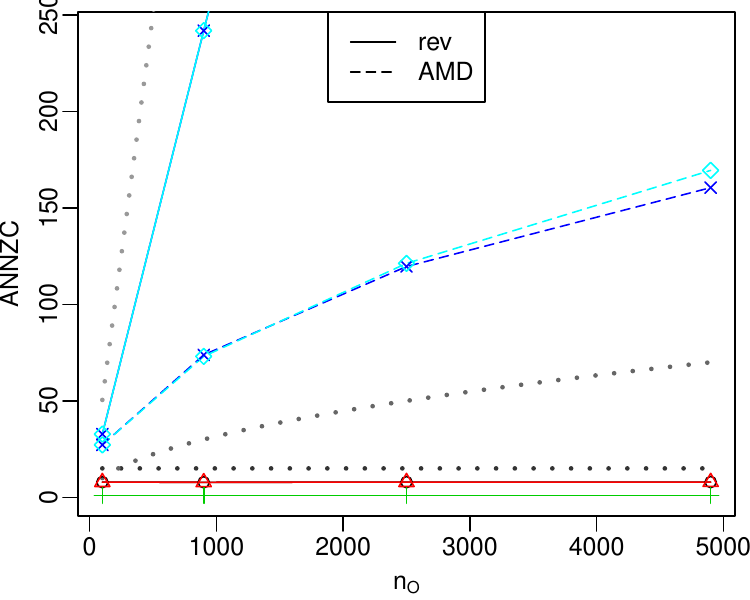}
	\caption{block}
	\end{subfigure}
  \caption{Average number of nonzero entries per column (ANNZC) in $\bV = \rchol(\bW)$ with $m=15$ neighbors for increasing $n_O=n_P$ on a unit square using maxmin ordering. 
  (a): Obtained using ``brute-force'' Cholesky based on reverse row-column ordering of $\bW$, resulting in a dense upper triangle of $\bV$ with ANNZC $=(n_O+1)/2$ for RF-full and LF-full. 
  (b): Obtained using \eqref{eq:Urf} or \eqref{eq:Vlf}, with ANNZC $<m$ for the RF methods, and some reduction of nonzero entries using approximate-minimum-degree (AMD) ordering of $\bW_{oo}$ for the LF methods}
\label{fig:Vsparsity}
\end{figure}

\subsubsection{Consistent framework\label{sec:consistent}}

Under OP ordering, the likelihood $\adens(\bz_o) = \int \adens(\bx) d\by$ is unchanged when $\by_p$ is removed from $\bx$, and the Vecchia approximation is applied to the resulting vector $\bx_{\tilde o} = \bx \setminus \by_p$:
\begin{equation}
\label{eq:impliedlik}
\textstyle \adens(\bz_o) = \int \int \adens(\bx_{\tilde o}) \adens(\by_p|\bx_{\tilde o}) d\by_p \, d\by_o = \int \adens(\bx_{\tilde o}) \int \adens(\by_p|\bx_{\tilde o}) d\by_p \, d\by_o = \int \adens(\bx_{\tilde o}) d\by_o.
\end{equation}
Thus, likelihood inference can first be carried out based on $\bx_{\tilde o}$ (i.e., only based on $\by_o$ and $\bz_o$) using \eqref{eq:likelihood} as described in \citet{Katzfuss2017a}. Then, $\by_p$ can be appended at the end of $\bx$ when predictions are desired, without changing the distribution $\adens(\bz_o)$. This has the advantage that parameter inference and prediction can be carried out in a consistent framework as described in Appendix \ref{sec:predictionsunknown}, and the prediction locations do not need to be known when training the model. If $\bU$ has already been calculated for $\bx_{\tilde o}$, it is also possible to reuse this matrix and simply append to it the columns corresponding to $\by_p$. Similarly, if prediction at additional locations is desired later, these can be ordered last and the existing matrices can be augmented, ensuring that the distribution of existing variables is unchanged.

For the response-first methods, the likelihood reduces to the standard \citet{Vecchia1988} likelihood. However, if only predictions are desired, we can set $g(i) = \emptyset$ for $i=1,\ldots,n_O$ without changing the approximation of $\dens(\by|\bz_o)$, resulting in computational savings. To see this, note that this choice only affects $\bU_{rr}$, which does not appear in any of the quantities in Section \ref{sec:predictions}, because $\bV = \bU_{\ell\ell}$ and $\bfmu=-(\bU_{r\ell}\bU_{\ell\ell}^{-1})'\bz_o$ for response-first.

Strictly speaking, LF-auto does not obey the OP restriction and hence has the undesirable property that changing $\locs$ (e.g., by adding prediction locations) might change the joint distribution of other variables in $\bx$. However, as shown in Figures \ref{fig:1Dillus} and \ref{fig:KL1D}, the LF-auto approximation is so accurate in one dimension that, even for small $m$, there is little difference to the exact GP, and so all joint distributions are almost identical to the exact ones.

\section{Simulation study \label{sec:comparison}}

We carried out a numerical comparison of the methods in Table \ref{tab:methods}. This systematic comparison was enabled by our R package \texttt{GPvecchia}, which implements all of the methods as special cases of the general Vecchia framework.

We simulated datasets at locations $\locs = \locs_o \cup \locs_p$, consisting of randomly drawn locations $\locs_o$ from an independent uniform distribution on $\domain$, combined with an equidistant grid $\locs_p$ on $\domain$. We simulated $\by$ at $\locs$ from the true distribution $\dens(\by)$ induced by a GP with Mat\'ern covariance function with variance 1 and smoothness parameter $\nu$, and then we sampled data $\bz_o$ by adding independent Gaussian noise with constant variance $\tau^2$ to $\by_o$. We call $1/\tau^2$ the signal-to-noise ratio (SNR). Effective range is the distance at which the correlation is 0.05.

Our simulation study focuses on the approximation accuracy of summaries of $\adens(\by|\bz_o)$, assuming that any potential hyperparameters $\bftheta$ are fixed and known. This results in more precise statements regarding the distinctions between the different methods, and avoids confounding with issues that are not the focus of our study, such as choice of inference algorithms, tuning parameters, or choice of prior distributions for $\bftheta$.

We computed the KL divergence for the joint distributions $\adens(\by_p|\bz_o)$ and $\adens(\by_o|\bz_o)$, and averaged over the results from multiple simulations for each method. This approximates the KL divergence for which the expectation is taken with respect to the joint distribution of the observations $\bz_o$ (see Section \ref{sec:approxacc}) and the observation locations $\locs_o$. We also computed the average marginal KL divergences for $\adens(\by_{i}|\bz_o)$ for $i \in p$ and for $i \in o$.

For ease of presentation, comparisons to the MRA and to an extension of sparse general Vecchia \citep{Katzfuss2017a} are omitted here and shown in Section \ref{sec:addcomp} instead.

\subsection{Numerical comparison in 1-D \label{sec:sim1D}}

First, we considered the unit interval, $\domain=[0,1]$, with $n_O=n_P=100$, effective range $0.15$, and 40 repetitions of the simulation. All methods used coordinate (left-to-right) ordering. To avoid numerical error when computing the KL divergence due to finite machine precision, we constrained the locations in $\locs_o$ to be at least $10^{-4}$ units apart. 
LF-auto is exact for $\nu=0.5$ with any $m \geq 1$. As shown in Figure \ref{fig:KL1D}, the method was much more accurate at the prediction locations than any of the other approaches for $\nu=1.5$. RF-ind and LF-ind, which do not condition on prediction locations, 
could only achieve a certain level of accuracy, with the joint KL divergence for the prediction locations leveling off as $m$ increased. The performance of RF-ind and LF-ind improved on the marginal measures. 

\begin{figure}[tp!]
\centering\includegraphics[width =.9\linewidth]{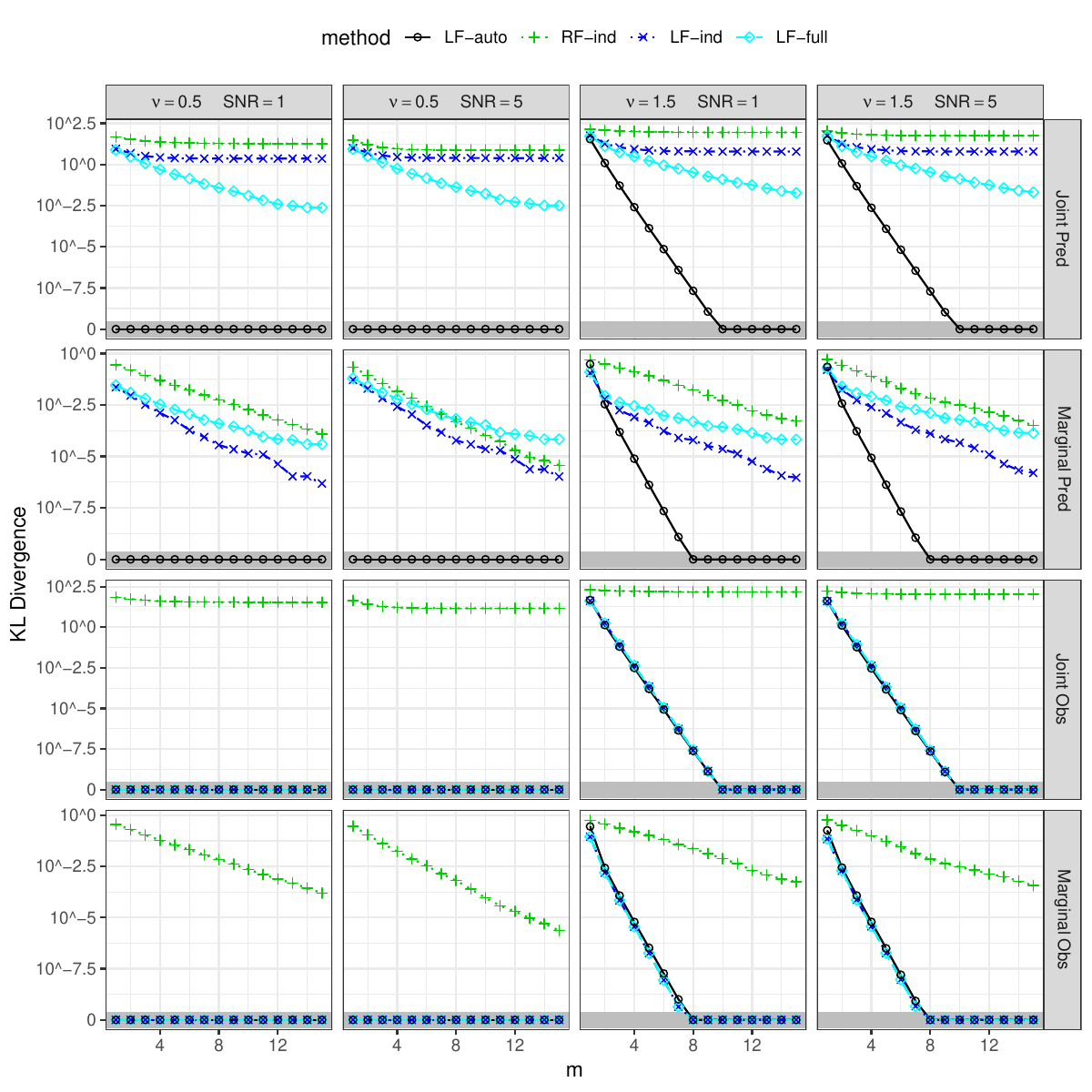}
  \caption{For a GP in \textbf{one dimension}, KL divergences as a function of conditioning-set size $m$. Rows correspond to average KL divergence for $\adens(\by_p|\bz_o)$ (Joint Pred), $\adens(\by_o|\bz_o)$ (Joint Obs), $\adens(\by_{i}|\bz_o)$ for $i \in p$ (Marginal Pred), and $\adens(\by_{i}|\bz_o)$ for $i \in o$ (Marginal Obs), respectively.
  We used a modified log scale for the y-axes, with values below $10^{-10}$ treated as zero and indicated by grey bars at the bottom.}
\label{fig:KL1D}
\end{figure}

\subsection{Numerical comparison in 2-D}

On the unit square, $\domain = [0,1]^2$, we used $n_O=n_P=4{,}900$ and effective range $0.15$, averaging over 20 repetitions. All methods used maxmin ordering. Figure \ref{fig:KL2D} shows the results of the simulations. LF-full is not computationally scalable (see Section \ref{sec:timingcomp}) but performed best in terms of accuracy, because it conditions only on latent variables (which contain more information about the process of interest than the noisy response variables). RF-full performed well on both joint and marginal accuracy measures.
For approaches using independent conditioning for the prediction locations (RF-ind and LF-ind), the joint KL divergence at the prediction locations did not converge to zero, but these methods were more competitive with the other methods on marginal measures, as expected. 

\begin{figure}[tp!]
\centering\includegraphics[width =.9\linewidth]{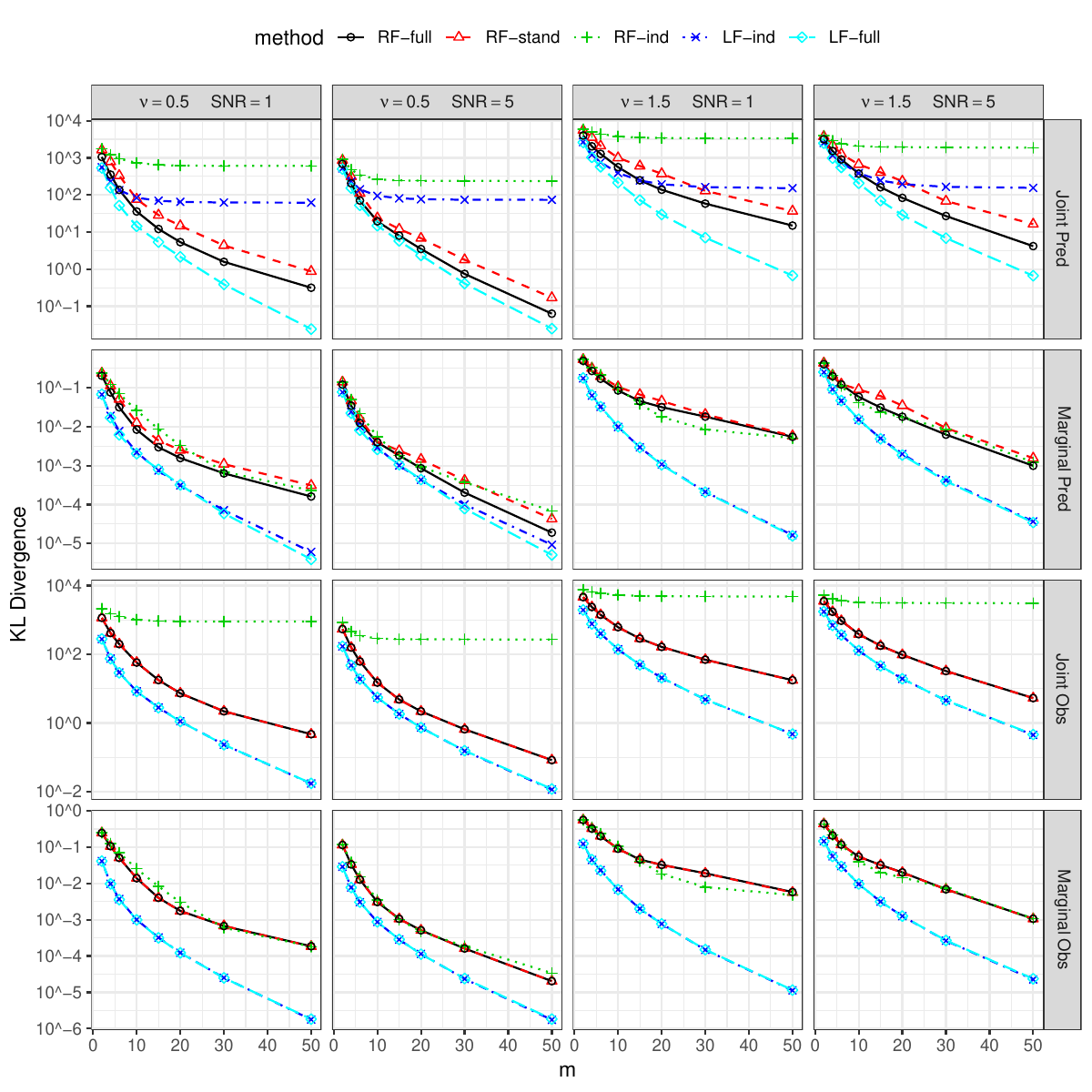}
  \caption{For a GP in \textbf{two dimensions}, KL divergences (on a log scale) as a function of conditioning-set size $m$. Rows correspond to average KL divergence for $\adens(\by_p|\bz_o)$ (Joint Pred), $\adens(\by_o|\bz_o)$ (Joint Obs), $\adens(\by_{i}|\bz_o)$ for $i \in p$ (Marginal Pred), and $\adens(\by_{i}|\bz_o)$ for $i \in o$ (Marginal Obs), respectively.}
\label{fig:KL2D}
\end{figure}

\subsection{Timing comparison \label{sec:timingcomp}}

We also carried out a timing study that examined the time for computing $\bU$ and $\bV$ on a unit square. Figure \ref{fig:timing} shows median computation times from five repetitions on a 4-core machine (Intel Core i7-3770) with 3.4GHz and 16GB RAM.
Consistent with our theoretical results, the time for computing $\bU$ increased roughly linearly with $n$, and was similar for all methods for given $n$ and $m$. (The time was slightly longer for the RF methods for small $m$, but this is solely due to an inefficiency in our RF code.) For the response-first methods, the time for computing $\bV$ was negligible relative to that for computing $\bU$. For LF-ind and LF-full, the time for computing $\bV$ using a fill-reducing permutation increased roughly between $\order(n_O^{3/2})$ and $\order(n_O^2)$, and the computation failed for large $n_O$ due to memory limitations. Computing $\bV_{oo}$ based on reverse ordering was even slower (see Figure \ref{fig:timing_noperm}).

\begin{figure}[tp!]
	\centering
	\begin{subfigure}{.48\textwidth}
	\centering
	\includegraphics[width =.97\linewidth]{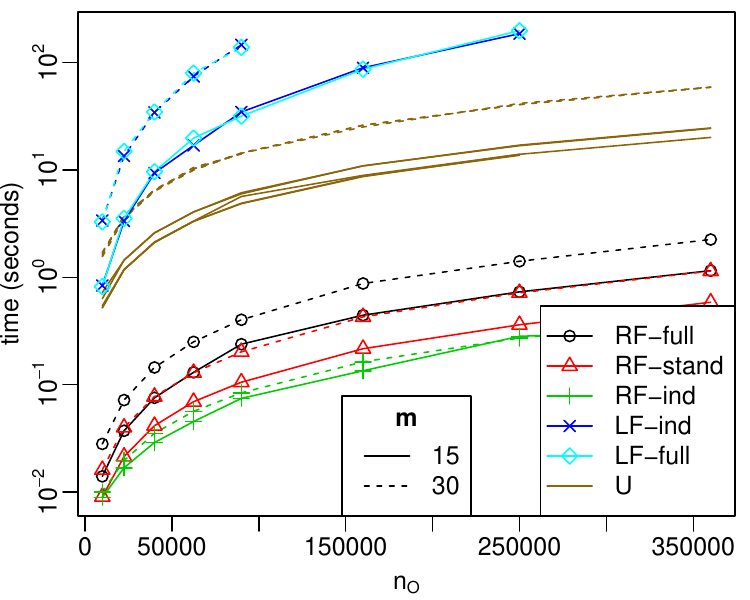}
	\caption{log scale}
	\end{subfigure}%
\hfill
	\begin{subfigure}{.48\textwidth}
	\centering
	\includegraphics[width =.97\linewidth]{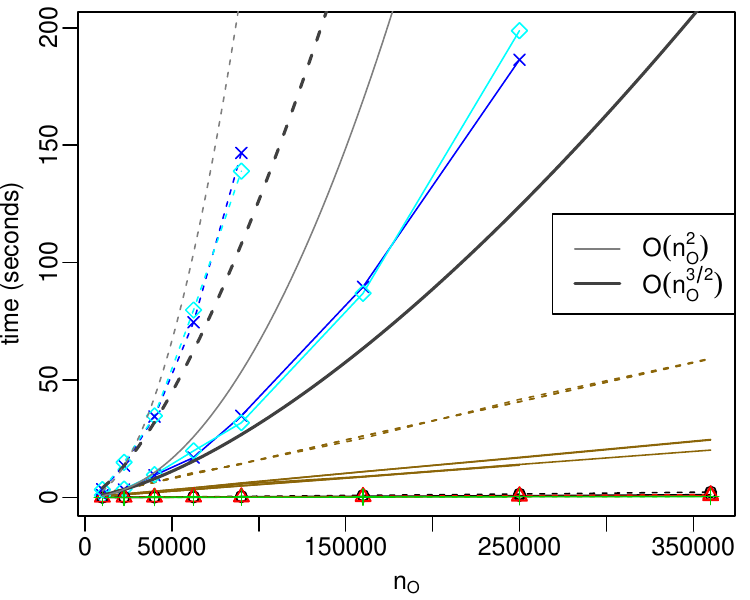}
	\caption{original scale}
	\end{subfigure}
  \caption{Time for computing $\bU$ and $\bV$ for $n_O = n_P$ observed and prediction locations on a unit square, as a function of $n_O$. Time for computing $\bU$ was similar for all methods. For RF methods, time for computing $\bV$ using \eqref{eq:Urf} was negligible. For LF methods, $\bV$ was computed using \eqref{eq:Vlf} and using the approximate-minimum-degree fill-reducing permutation for $\bW_{oo}$.}
\label{fig:timing}
\end{figure}

\subsection{Comparison for large $n$\label{sec:largen}}

We further compared the scalable response-first methods for large $n_O = n_P$, with smoothness $\nu=0.5$ and effective range $0.15$ on a unit square. Two modifications to previous comparisons were necessary due to the large data size. First, we simulated the GP values on a regular $1{,}000 \times 1{,}000$ grid, using a regular subgrid of size $n_P$ as $\locs_p$ and subsampling $n_O$ of the remaining grid points as $\locs_o$. Second, as it was impossible to compute the exact KL divergence, we approximated it by subtracting $\log \adens(\by|\bz_o)$ for each method from $\log \adens(\by|\bz_o)$ as approximated by a ``very accurate'' RF-full model with $m=60$, all averaged over ten simulated datasets. As shown in Figure \ref{fig:large_sample}, RF-full was much more accurate than the other methods in all settings.

\begin{figure}[tp!]
	\centering
	\begin{subfigure}{.63\textwidth}
	\centering
	\includegraphics[width =.97\linewidth]{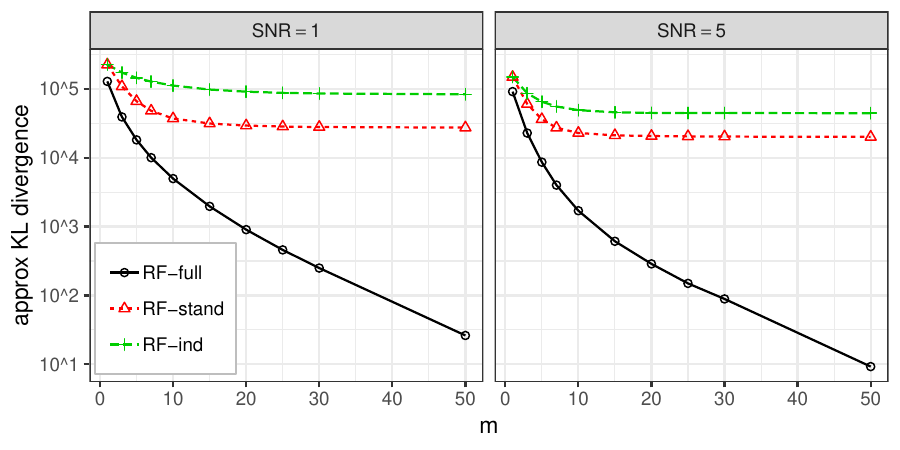}
	\caption{$n_O = n_P = 90{,}000$}
	\end{subfigure}%
\hfill
	\begin{subfigure}{.36\textwidth}
	\centering
    	\includegraphics[width =.98\linewidth]{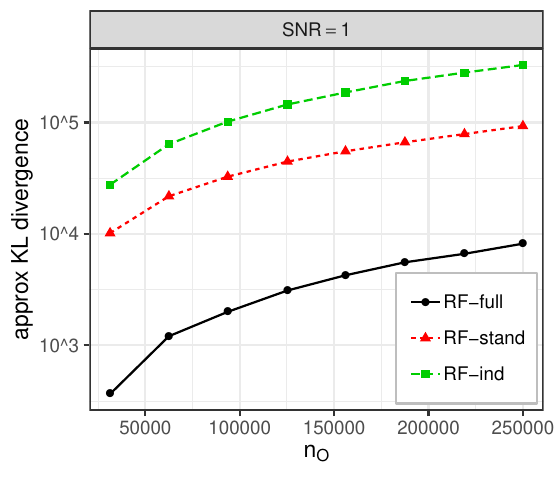}
	\caption{$m=15$}
	\end{subfigure}
  \caption{Approximate KL divergences for the joint distribution $\adens(\by|\bz_o)$ in 2D with smoothness $\nu = 0.5$. (a) Fixed $n_O = n_P = 90{,}000$, varying $m$. (b) Fixed $m=15$, increasing $n_O=n_P$ from $30{,}000$ to $250{,}000$ under in-fill (fixed-domain) asymptotics}
\label{fig:large_sample}
\end{figure}

\subsection{Heaton comparison\label{sec:heaton}}

We compared the response-first methods, all with conditioning-set size $m=15$, to the methods considered in a recent review and comparison paper \citep{Heaton2017}. We used the data from \citet[][Sec.~4.1]{Heaton2017}, which consist of $n_O=105{,}569$ training data and $n_P=44{,}431$ test data simulated from a GP with exponential covariance. For our methods, we estimated the mean as the sample average of the training data, and then estimated the process variance, noise variance, and range parameter (with true values 16.4, 0.05, and 4/3, respectively) by maximizing the likelihood approximated by sparse general Vecchia \citep{Katzfuss2017a} on a randomly chosen training subset of size 10,000. 

We compared to the results of the top five methods reported in \citet[][Tab.~2]{Heaton2017}, based on the RMSE, the continuous rank probability score (CRPS), and the computation time. Timing results in \citet[][]{Heaton2017} were obtained on the Becker computing
environment at Brigham Young University; the times for maximizing the likelihood and computing  predictions for our methods were obtained on a basic desktop computer (Intel Core i5-3570 CPU @ 3.4GHz), ignoring some set-up costs.
Only marginal predictions were considered in \citet{Heaton2017}; for our methods, we also computed the average log score for 10 randomly selected subsets of size 500 of the test data.
CRPS and log score are proper scoring rules that evaluate the approximation error in the predictive distribution \citep[e.g.,][]{Gneiting2014}, and simultaneously reward accurate point prediction (i.e., posterior mean) and accurate uncertainty quantification.

\begin{table}[tbp]
\centering
\small
\begin{tabular}{l|l|rrrr}
 & Method & RMSE & CRPS & JLS & Time (min)\\ 
  \hline
\multirow{3}{*}{$m=15$} & RF-full & \textbf{0.82} & \textbf{0.43} & \textbf{368.5} & 0.47 \\ 
  & RF-stand & 0.83 & \textbf{0.43} & \textbf{368.5} & 0.45 \\ 
  & RF-ind & 0.87 & 0.45 & 784.7 & 0.41 \\ 
  \hline
\multirow{5}{*}{Heaton} &   NNGP ($m\!=\!15$) & 0.88 & 0.46 & & 1.99\\     
 & MRA ($m\!>\!500$) & 0.83 & \textbf{0.43} & & 13.57\\  
 & LatticeKrig & 0.87 & 0.45 & & 25.58\\
 & Partition  & 0.86 & 0.47 & & 77.56\\  
 & SPDE & 0.86 & 0.59 & & 138.34
\end{tabular}
\caption{Comparison using data from \citet[][Sec.~4.1]{Heaton2017}, with lowest (i.e., best) scores in bold. First three rows: Our methods with conditioning-set size $m=15$. Bottom five rows: Results for the five best methods taken directly from \citet[][Tab.~2]{Heaton2017}. CRPS: continuous rank probability score. JLS: Joint log score for test subsets.} 
\label{tab:heaton}
\end{table}

The results are shown in Table \ref{tab:heaton}. RF-full and RF-stand had similar scores, because the assumed noise level was negligible. Both approaches outperformed all other methods; only MRA achieved comparable scores, but required much larger $m$ and computation time. (A separate, more thorough comparison to MRA can be found in Section \ref{sec:addcomp}.) The NNGP results reported in \citet{Heaton2017} were obtained using a variant called NNGP-conjugate in \citet{Finley2017}, which can be viewed as a Bayesian version of RF-ind; indeed, the marginal scores for the two methods were very similar. The joint log score for RF-ind was much worse than for RF-full and RF-stand.

\section{Application to satellite data \label{sec:realdata}}

We applied the scalable response-first methods from Table \ref{tab:methods} to Level-2 bias-corrected solar-induced chlorophyll fluorescence (SIF) retrievals over land from the Orbiting Carbon Observatory 2 (OCO-2) satellite \citep{OCO2chlorophyll2015}. 
SIF is an important proxy for the amount of biomass produced from photosynthesis \citep{sun2017oco,sun2018overview} and can be used to monitor the productivity of crops \citep{guan2016improving}. The OCO-2 satellite has a sun-synchronous orbit with a period of 99 minutes and approximately repeats its spatial coverage every 16 days. Many remote-sensing satellites follow a similar sun-synchronous orbit, an orbital pattern that produces very dense observations along the trajectory of the orbit but wide gaps in space or time between orbits. This pattern is very common, and it presents a challenge for existing nearest-neighbor-based prediction methods, because the nearest neighbors to any point in space or time will almost always be a sequence of densely packed points along one of the orbits. As we show below, this choice can be suboptimal and produce unrealistic artifacts in predicted maps.

We analyzed chlorophyll fluorescence data collected between August 1 and August 31, 2018 over the contiguous United States. 
During this time period, there were a total of 245{,}236 observations, plotted in Figure \ref{fig:fluordata}. There was little evidence of temporal change during the time period, so we restricted our attention to a purely spatial model. We modeled the data with the spatial Gaussian process
\[
\textstyle z(\bs_i) = \beta_0 + \sum_{j=1}^p \beta_j X_j(\bs_i) + y(\bs_i) + \epsilon_i, \qquad \bs_1,\ldots,\bs_n \in \domain,
\]
where $X_j$ were Gaussian basis functions centered at knots chosen as the first $p=50$ locations in a maxmin ordering of the data locations, which ensured that no two knots were placed close to each other. The basis range was selected to be 637km (10\% of Earth radius). The basis functions were included to capture a large amount of unstructured long-range variability that could not be explained by simple linear functions of latitude and longitude. The parameters $\beta_0,\ldots,\beta_p$ were estimated using least squares.

The residual field (shown in Figure \ref{data_and_residuals}) was modeled as $y(\cdot) \sim \GP(0,K)$, where $K$ was assumed to be an isotropic Mat\'ern covariance function with three parameters: variance, range, smoothness. The noise terms $\epsilon_i$ were assumed to be independent and identically distributed as $N(0,\tau^2)$. For covariance-parameter estimation on the residuals, we used the sparse general Vecchia likelihood \citep{Katzfuss2017a} with maxmin ordering, increasing $m$ up to 40, beyond which the estimates did not change significantly. The estimated parameters were: variance = 0.1097, range = 100.8 km, smoothness = 0.0982, noise variance $\widehat{\tau}^2 = 0.1869$. 

\begin{figure}[tp!]
\centering
\includegraphics[width=.9\textwidth]{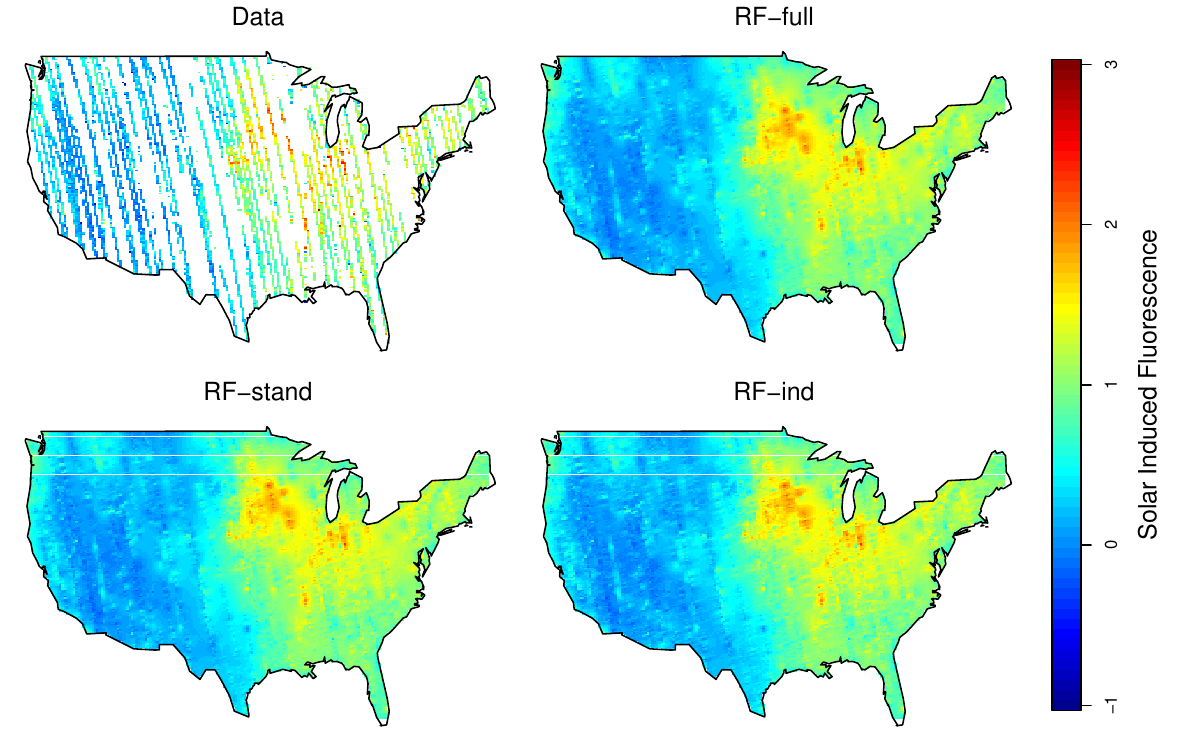}
\caption{\label{fig:fluordata}Satellite data, together with Vecchia predictions using $m=30$ neighbors}
\end{figure}

Using the estimated covariance function and noise variance, we computed predictions for the scalable methods RF-full, RF-stand, and RF-ind (local kriging) at a grid of size $n_P = 24{,}407$ over the contiguous United States. Figure \ref{fig:fluordata} shows predictions (i.e., posterior means of $\beta_0 + \sum_{j=1}^p \beta_j X_j(\bs) + y(\bs)$) for $m=30$ neighbors. Because the observations cover the study region quite well, the predictions using the various methods look similar, as might also be expected from the second row of Figure \ref{fig:KL2D}. Upon closer inspection, however, the RF-ind predictions appear noisier and exhibit a ``streaky'' behavior. Figure \ref{fig:fluordatatexas} shows predictions at a higher resolution of $n_P = 18{,}576$ locations over the state of Texas. We can see that the RF-ind estimates are noisier and have more pronounced discontinuities parallel and perpendicular to the swaths of data. We conjecture that because the data locations are so dense along each swath, for RF-ind, two nearby prediction locations can condition on entirely different sets of observations if the two locations are nearly equidistant from two different swaths. Section \ref{sec:satsupp} contains additional plots showing prediction uncertainties (i.e., posterior standard deviations) for $m=30$ and predictions with $m=60$. For $m=60$, all predicted maps appeared smoother, but the Texas maps still had clearly visible streaks for RF-ind.

\begin{figure}[tp!]
\centering
\includegraphics[width=0.9\textwidth]{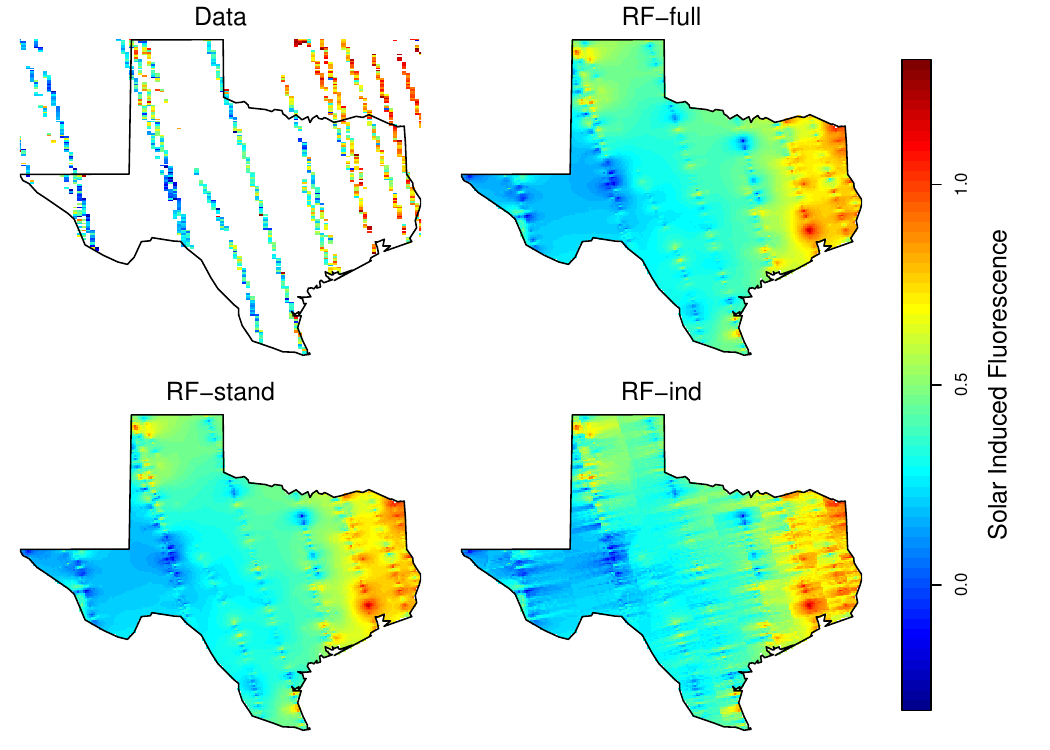}
\caption{\label{fig:fluordatatexas}Satellite data, together with Vecchia predictions over Texas using $m=30$ neighbors}
\end{figure}

We also compared the prediction accuracy of RF-full, RF-stand, and RF-ind in two cross-validation experiments. First, we selected 10 separate prediction sets, each consisting of 4000 randomly sampled data locations, to evaluate short-range prediction. Second, we held out 10 of the 67 swaths, one at a time, to evaluate long-range prediction. The average held-out swath size was 3{,}493 locations. For each of the two cross-validation experiments, we computed the root mean squared error (RMSE) and the total log score, obtained by summing the negative log predictive densities for each held-out test set. The log scores are reported relative to the lowest-achieved log score. 

The resulting prediction scores are shown in Figure \ref{fig:fluorcomp}. For all settings, RF-full performed best. 
RF-ind (i.e., local kriging) was not competitive in terms of long-range predictions on the swath test sets. For the random test sets, RF-stand and local kriging performed similarly, with both methods roughly requiring $m=20$ to achieve the same accuracy as RF-full with $m=10$. These differences can be substantial in terms of computation times, which scale cubically in $m$.
While the absolute differences in the RMSE values were not large, this was at least partially due to all comparisons being carried out relative to the (very noisy) test data $\bz_p$, as the true fluorescence values $\by_p$ are unknown. The ``convergence'' of RF-full, with similar values for $m=20$ as for $m=40$, indicates that even the exact GP without any approximation likely would not achieve significantly lower RMSE.
Differences in the log scores were more pronounced, indicating that our new RF-full method can substantially outperform local kriging in terms of uncertainty quantification.

\begin{figure}
\centering 
\includegraphics[trim=0mm 4mm 0mm 3mm, clip, width=\textwidth]{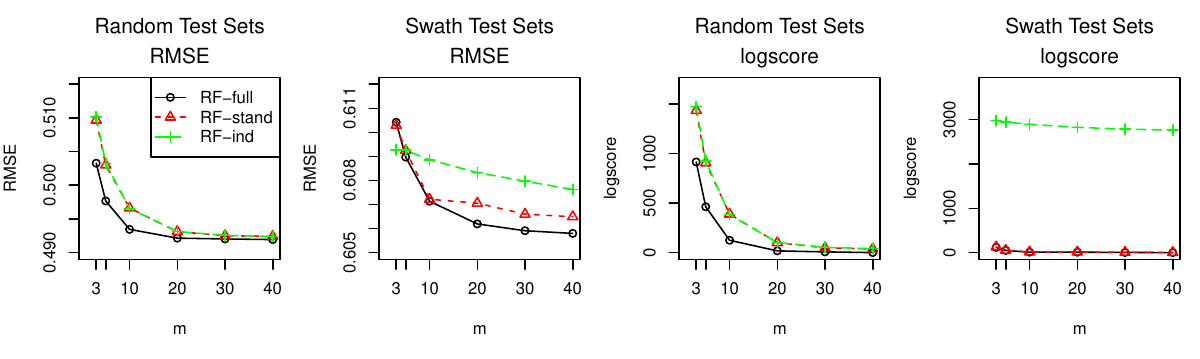}
\caption{\label{fig:fluorcomp} For fluorescence data, comparison of prediction scores as a function of conditioning-set size $m$}
\end{figure}

\section{Conclusions \label{sec:conclusions}}

Vecchia approximation of Gaussian processes (GPs) is a powerful computational tool for fast analysis of large spatial datasets. While Vecchia approximations have been very popular for likelihood approximations, their use for the very important task of GP prediction or kriging had not been fully examined. Here, we proposed a general Vecchia framework for GP predictions, which includes as special cases some existing and several novel computational approaches. We studied the accuracy and computational properties of the prediction methods both theoretically and numerically. In the case of unknown hyperparameters, all methods can be extended straightforwardly as described in Appendix \ref{sec:predictionsunknown}.

Based on our results, we make the following recommendations, which are also summarized briefly in Table \ref{tab:methods}. On a one-dimensional domain, LF-auto clearly had the best performance in all of the settings we considered. The auto-regressive structure in LF-auto also affords linear computational scaling, and so we recommend LF-auto without any qualifications when the domain is one-dimensional. 
In two dimensions, we generally recommend RF-full, as it scales linearly and performed well on all accuracy measures. RF-stand is less accurate for noisy data, but has some computational advantages when the number of prediction locations is much smaller than the number of observations. LF-full can be very accurate, but it does not scale linearly in the data size.
Local kriging (RF-ind) is fast and can provide accurate marginal predictive distributions, but it ignores dependence in the joint predictive distributions. LF-ind does not scale linearly and its joint predictive distributions at unobserved locations were often less accurate than those from RF-full.
These inferential limitations are evident in Figure \ref{fig:1Dillus} and in the top rows of Figures \ref{fig:KL1D} and \ref{fig:KL2D}.

The methods and algorithms proposed here are implemented in the R package \texttt{GPvecchia} \citep{GPvecchia}, with default settings reflecting the recommendations in the previous paragraph. In principle, our methods and code are applicable in more than two dimensions, but a thorough investigation of their properties in this context is warranted. For example, a follow-up paper \citep{Katzfuss2020} shows that Vecchia-based approximations, using our observed-prediction ordering and appropriate extensions, can be highly accurate for computer-model emulation in higher dimensions.
A further follow-up paper \citep{Zilber2019} extends our methods to Vecchia-Laplace approximations of generalized GPs for non-Gaussian spatial data.

\footnotesize
\appendix
\section*{Acknowledgments}

Katzfuss' research was partially supported by National Science Foundation (NSF) grant DMS--1521676 and NSF CAREER grant DMS--1654083.
Guinness' research was partially supported by NSF grant DMS--1613219 and NIH grant No.\ R01ES027892. The authors would like to thank Anirban Bhattacharya, David Jones, Jennifer Hoeting, and several anonymous reviewers for helpful comments and suggestions. Jingjie Zhang and Marcin Jurek contributed to the R package \texttt{GPvecchia}, and Florian Sch\"afer provided C code for the exact maxmin ordering.

\section{Vector and indexing notation\label{app:notation}}

As an example, define $\by = (\by_1,\by_2,\by_3,\by_4,\by_5)$ as a vector of vectors. Vectors of indices are used for defining subvectors. For example if $o = (4,1,2)$ is a vector of indices, then $\by_o = (\by_4,\by_1,\by_2)$. Unions of vectors are vectors and are defined when the two vectors have the same type and when the ordering is specified. For example, if $\bz = (\bz_1,\bz_2)$, then $\by_o \cup \bz = (\by_4,\by_1, \bz_1, \by_2, \bz_2)$ defines the union of $\by_o$ and $\bz$. Intersection is defined similarly and uses the $\cap$ notation.
The ordering of the elements of the union or intersection can be defined alternatively via an index function $\#$ taking in an element and a vector and returning the index occupied by the element in the vector. Continuing the example above, $\#(\by_4, \by) = 4$, whereas $\#(\by_4, \by_o \cup \bz) = 1$. 
A full description of the vector and indexing notation can be found in \citet[][App.~A]{Katzfuss2017a}.

\section{Computing $\bU$ \label{app:computeU}}

We recapitulate here the formulas for computing $\bU$ from \citet{Katzfuss2017a}. Let $g(i)$ denote the vector of indices of the elements in $\bx$ on which $x_i$ conditions.
Also define $C(x_i,x_j)$ as the covariance between $x_i$ and $x_j$ implied by the true model in Section \ref{sec:GP}; that is, $C(y_i,y_j) = C(z_i,y_j) = K(\bs_i,\bs_j)$ and $C(z_i,z_j) = K(\bs_i,\bs_j) + \indicat_{i=j} \tau^2_i$. Then, the $(j,i)$th element of $\bU$ can be calculated as
\begin{equation}
\label{eq:U}
\bU_{ji} = \begin{cases} d_i^{-1/2}, & i=j,\\ -b_{i}^{(j)} d_i^{-1/2}, & j \in g(i), \\ 0, &\textnormal{otherwise}, \end{cases}
\end{equation}
where $\bb_i'= C(x_i,\bx_{g(i)}) C(\bx_{g(i)},\bx_{g(i)})^{-1}$, $d_i = C(x_i,x_i) - \bb_i' C(\bx_{g(i)},x_i)$, and $b_i^{(j)}$ denotes the $k$th element of $\bb_i$ if $j$ is the $k$th element in $g(i)$ (i.e., $b_i^{(j)}$ is the element of $\bb_i$ corresponding to $x_j$).

\section{Prediction with unknown parameters\label{sec:predictionsunknown}}

In practice, most GP models depend on an unknown parameter vector $\bftheta$, which we will make explicit here. Likelihood approximation for parameter inference is discussed in detail in \citet{Katzfuss2017a}, but we will review it briefly here. 
Integration of $\adens(\bx|\bftheta)$ in \eqref{eq:gvp} with respect to $\by$ results in the following Vecchia likelihood \citep[][Prop.~2]{Katzfuss2017a}:
\begin{equation}
\label{eq:likelihood}
\textstyle -2 \log \adens(\bz_o|\bftheta) = -2 \sum_{i} \log \bU_{ii} + 2\sum_{i} \log \bV_{ii} + \tilde\bz'\tilde\bz - (\bV^{-1}\bU_{\ell,\all}\tilde\bz)'(\bV^{-1}\bU_{\ell,\all}\tilde\bz) + n \log(2\pi),
\end{equation}
where $\bU$ and $\bV$ implicitly depend on $\bftheta$, and $\tilde\bz = \bU_{r,\all}'\bz_o$.
The computational cost for evaluating this Vecchia likelihood is often low, and \citet{Katzfuss2017a} provide conditions on the $g(i)$ under which the cost is guaranteed to be linear in $n$.

This allows for various forms of likelihood-based parameter inference. 
In a frequentist setting, we can compute $\hat\bftheta = \argmax_{\bftheta} \log \adens(\bz_o|\bftheta)$, and then compute summaries of the posterior predictive distribution $\adens(\by|\bz_o) = \normal_n(\bfmu(\hat\bftheta),\bfSigma(\hat\bftheta))$ as described in Section \ref{sec:predictions}. This often has low computational cost but ignores uncertainty in $\hat\bftheta$. An example is given in Section \ref{sec:realdata}.

In a Bayesian setting, given a prior distribution $\dens(\bftheta)$, a Metropolis-Hastings sampler can be used for parameters whose posterior or full-conditional distribution are not available in closed form. At the $(l+1)$th step of the algorithm, one would propose a new value $\bftheta^{(P)} \sim q(\bftheta|\bftheta^{(l)})$ and accept it with probability $\min(1,h(\bftheta^{(P)},\bftheta^{(l)})/h(\bftheta^{(l)},\bftheta^{(P)}))$, where $h(\bftheta,\tilde\bftheta) = \dens(\bftheta)\adens(\bz_o|\bftheta) q(\tilde\bftheta|\bftheta)$. After burn-in and thinning, this results in a sample, say, $\bftheta^{(1)},\ldots,\bftheta^{(L)}$, leading to a Gaussian-mixture prediction: $\adens(\by|\bz_o) = (1/L) \sum_{l=1}^L \normal_n(\bfmu(\bftheta^{(l)}),\bfSigma(\bftheta^{(l)}))$. See \citet[][App.~E]{Katzfuss2017a} for an example for the use of RF-full in this setting. For more complicated Bayesian hierarchical models, inference can be carried out using a Gibbs sampler in which $\by$ is sampled from its full-conditional distribution as described in item 4.\ in Section \ref{sec:predictions}.

\section{Numerical nonzeros in $\bV$\label{app:nonzeros}}

For simplicity, we focus here on RF-full, although numerical nonzeros can similarly occur for RF-stand, LF-full, and LF-ind (see Figure \ref{fig:V2D}). For RF-full, the upper triangle of $\bW$ is at least as dense as the upper triangle of $\bV$.
Specifically, for $j<i$, $\bV_{ji}=\bU_{\ell_j\ell_i}=0$ unless $j \in q_y(i)$.
From \citet[][Prop.~3.2]{Katzfuss2017a}, we have that $\bW_{ji} =0$ unless $j \in q_y(i)$ or $\exists k > i$ such that $i,j \in q_y(k)$. Thus, for any pair $j<i$ such that $j \notin q_y(i)$ but $i,j \in q_y(k)$ for some $k>i$, we generally have $\bV_{ji} = 0$ and $\bW_{ji} \neq 0$.

From the standard Cholesky algorithm, we can derive that the algorithm for $\bV = \rchol(\bW)$ computes $\bV_{ji}$ as
\[
\textstyle\bV_{ji} = (\bW_{ji} - \sum_{k=i+1}^n\bV_{ik}\bV_{jk})/\bV_{ii}.
\]
Thus, for any pair $j<i$ such that $\bV_{ji}=0$ but $\bW_{ji} \neq 0$, we know that $\bW_{ji} = \sum_{k=i+1}^n\bV_{ik}\bV_{jk}$ theoretically, but due to potential numerical error it is not guaranteed that this equation holds exactly. A numerical nonzero is introduced in $\bV$ whenever a rounding error occurs in $\sum_{k=i+1}^n\bV_{ik}\bV_{jk}$, which relies on (potentially many) previous calculations in the Cholesky algorithm. Such numerical nonzeros are avoided by extracting $\bV=\bU_{\ell\ell}$ as a submatrix of $\bU$ (as proposed in Section \ref{sec:responsefirstordering}), instead of explicitly carrying out the Cholesky factorization $\bV = \rchol(\bW)$.

\subsection{Implications for selected inverse \label{app:selinv}}

When $\bV$ is computed by copying a submatrix from $\bU$ to avoid numerical nonzeros, the selected inverse of this $\bV$ is not guaranteed to return the exact posterior variances of $\by$, unless $\bV$ is ``padded'' with zeros, which results in additional costs.
This is because the selected inverse operates on the symbolic nonzero elements; that is, it operates on all elements that have to be computed in the Cholesky, even if they cancel to zero numerically (which is the case for many entries in our case). Denoting by $\bS$ the selected inverse of $\bW$ based on $\bV$, a close look at the Takahashi recursions reveals that for all $j,k$ with $j,k \in q_y(i)$, we need $\bS_{ji}$ and $\bS_{kj}$. The latter element is only calculated if $j \in q_y(k)$. However, if $j \notin q_y(k)$, $\bS_{kj} =\cov(y_j,y_k|\bz)$ will typically be very small (if $m$ is reasonably large), because their corresponding locations will likely be far away from each other and data can be observed in between.
In our experiments, the additional approximation error introduced by the selected inverse was negligible relative to the error introduced by the Vecchia approximation itself. When $m$ becomes large enough for the Vecchia approximation to be accurate, the additional approximation error introduced by SelInv goes to zero as well.

If the exact variances implied by the Vecchia approximation $\adens(\by|\bz_o)$ are desired, they can be computed as $\diag(\var(\by_p|\bz_o)) = ((\bV^{-1}\bI_{\all p})\circ(\bV^{-1}\bI_{\all p}))'\bfone_{n}$. For RF-full and RF-stand, this requires $\order(nmn_P)$ time, as described in Section \ref{sec:complexity}, and so the overall computational complexity would not be increased if $n_P = \order(m^2)$.

\section{Proofs \label{app:proofs}}

In this section, we provide proofs for the propositions stated throughout the article.

\begin{proof}[Proof of Proposition \ref{prop:KLvecchia}]
Part 1 of this proposition is equivalent to Thm.~1 in \citet{Guinness2016a}. We prove the statement here again in a different way, which can be easily extended to prove Part 2.

First, consider a generic Vecchia approximation of a vector $\bx$ of length $n$: $\adens(\bx) = \prod_{i=1}^n \dens(x_i|\bx_{g(i)}) = \prod_{i=1}^n \normal(x_i|\mu_{i|g(i)},\sigma^2_{i|g(i)})$. Then,
\[
\textstyle (-2) \E^\bx \log \adens(\bx) = \sum_{i=1}^n \log \var(x_i|\bx_{g(i)}) + \sum_{i=1}^n \E^\bx w_i^2 + n \log(2\pi),
\]
where $w_i = (x_i - \mu_{i|g(i)})/\sigma_{i|g(i)}$. We have 
$\E(w_i) = 0$, because $\E(\mu_{i|g(i)}) = \E\E(x_i|\bx_{g(i)}) = \E(x_i)$. We also have $\var(w_i)=1$, because
\[
\var(x_i-\mu_{i|g(i)}) = \var \E\big(x_i-\E(x_i|\bx_{g(i)})|\bx_{g(i)}\big) + \var\big(x_i-\E(x_i|\bx_{g(i)})|\bx_{g(i)}\big) = 0 + \sigma^2_{i|g(i)}.
\]
Hence, $w_i \sim \normal(0,1)$, and so $\E^\bx w_i^2 = 1$ and 
\[
\textstyle (-2) \E^\bx \log \adens(\bx) = \sum_{i=1}^n \log \var(x_i|\bx_{g(i)}) + n + n \log(2\pi).
\]
Because the exact density $\dens(\bx)$ is a special case of Vecchia with $g(i) = (1,\ldots,i-1)$, we have
\begin{equation}
\label{eq:KLvecc}
\textstyle \KL\big(\dens(\bx)\|\adens(\bx)\big) = \E^\bx \log \dens(\bx) - \E^\bx \log\adens(\bx) = \frac{1}{2} \sum_{i=1}^n \log \frac{\var(x_i|\bx_{g(i)})}{\var(x_i|\bx_{1:i-1})}.
\end{equation}

\noindent
For $g_1(i) \subset g_2(i)$, we can write, say, $g_2(i) = g_1(i) \cup c(i)$. Using the law of total variance,
\begin{equation}
\label{eq:reducecondvar}
\var(x_i|\bx_{g_1(i)}) = \var(x_i|\bx_{g_1(i)}, \bx_{c(i)}) + \var(\E(x_i|\bx_{g_1(i)}, \bx_{c(i)})|\bx_{c(i)}) \geq \var(x_i|\bx_{g_2(i)}).
\end{equation}

\noindent
Now, Part 1 follows by combining \eqref{eq:KLvecc} and \eqref{eq:reducecondvar} with the assumption that $g_1(i) \subset g_2(i)$ for all $i$.

\noindent
For Part 2, we consider $\bx=(\bx^{(1)},\bx^{(2)})$ with $\bx^{(1)} = (x_1,\ldots,x_{n_1})$ and $\bx^{(2)} = (x_{n_1+1},\ldots,x_{n_1+n_2})$. Then,
\begin{align*}
\textstyle \CKL\big(\dens(\bx^{(2)}|\bx^{(1)})\|\adens(\bx^{(2)}|\bx^{(1)})\big) 
& \textstyle = \int \dens(\bx^{(1)}) \int \dens(\bx^{(2)}|\bx^{(1)}) \log\big(\dens(\bx^{(2)}|\bx^{(1)})/\adens(\bx^{(2)}|\bx^{(1)})\big) d\bx^{(2)} d\bx^{(1)} \\
& \textstyle = \E^\bx \log \dens(\bx^{(2)}|\bx^{(1)}) - \E^\bx \log \adens(\bx^{(2)}|\bx^{(1)}) \\
& \textstyle = \frac{1}{2} \sum_{i=n_1+1}^{n_1+n_2} \log \frac{\var(x_i|\bx_{g(i)})}{\var(x_i|\bx_{1:i-1})},
\end{align*}
where the last equality can be shown almost identically to \eqref{eq:KLvecc} above, noting that
$\adens(\bx^{(2)}|\bx^{(1)}) = \adens(\bx)/\adens(\bx^{(1)}) = \prod_{i=n_1+1}^{n_1+n_2} \dens(x_i|\bx_{g(i)})$. Part 2 follows by combining this result and \eqref{eq:reducecondvar} with the assumption that $g_1(i) \subset g_2(i)$ for all $i=n_1+1,\ldots,n_1+n_2$.
\end{proof}

\begin{proof}[Proof of Proposition \ref{prop:KLordering}]
For all parts of this proposition, note that all variables in the model are conditionally independent of $z_j$ given $y_j$, and so conditioning on $y_j$ is equivalent to conditioning on $y_j$ and $z_j$. 
For Part 1, we can thus verify easily that $g(i)$ is equivalent to $(1,\ldots,i-1)$, for all $i$ and all methods under consideration, and so Part 1 follows from \eqref{eq:KLvecc}. 
Using Proposition \ref{prop:KLvecchia}, the proof for all other parts simply consists of showing that certain conditioning vectors contain certain other conditioning vectors \citep[similar to][Prop.~4]{Katzfuss2017a}.
For example, for Part 3, all response-first methods are based on the ordering $\bx = (\bz_o,\by_o,\by_p)$ with nearest-neighbor conditioning (under some restrictions for RF-stand and RF-ind), and so it is easy to see that $g_m(i) \subset g_{m+1}(i)$ for all $i \in \ell_p=(n+1,\ldots,n+n_P)$. LF-full and LF-ind can equivalently be defined based on the ordering $(\by_o,\bz_o,\by_p)$, in which case we also have $g_m(i) \subset g_{m+1}(i)$ for all $i \in \ell_p=(n+1,\ldots,n+n_P)$.
For Part 5, note that in RF-stand and RF-ind, $\by_p$ does not condition on $\by_o$. Hence, the distribution $\adens(\by_p|\bz_o)$ is equivalent to the distribution obtained under a simplified Vecchia approximation based on $\bx=(\bz_o,\by_p)$ (i.e., with $\by_o$ removed completely). It is straightforward to show that Part 5 holds for this simplified approximation.
For Part 6, note that $g(i)$ is the same for RF-full and RF-stand for all $i=1,\ldots,n$. For $i \in p$, letting $a(i) = q_y^{\text{RF-full}}(i) \cap o$, we have $\bx_{g^{\text{RF-full}}(i)}=(\by_{q_y^{\text{RF-stand}}},\by_{a(i)})=(\by_{q_y^{\text{RF-stand}}},\by_{a(i)},\bz_{a(i)})$ and $\bx_{g^{\text{RF-stand}}(i)}=(\by_{q_y^{\text{RF-stand}}},\bz_{a(i)})$, and so $g^{\text{RF-stand}}(i) \subset g^{\text{RF-full}}(i)$ for all $i \in p$.
For Part 7, note that a GP with exponential covariance in 1-D is a Markov process, and so in \eqref{eq:lf} with $q_y(i)=i-1$ and left-to-right ordering, we have $\dens(y_i|\by_{q_y(i)})=\dens(y_i|y_{1},\ldots,y_{i-1})$ and hence $\adens(\bx) = \dens(\bx)$.
\end{proof}

\footnotesize
\bibliographystyle{apalike}
\bibliography{vecchiabib}

\newpage

\renewcommand{\thepage}{S\arabic{page}}  
\renewcommand{\thesection}{S\arabic{section}}   
\renewcommand{\thetable}{S\arabic{table}}   
\renewcommand{\thefigure}{S\arabic{figure}}
\renewcommand{\theequation}{S\arabic{equation}}

\setcounter{page}{1}
\setcounter{section}{0}
\setcounter{table}{0}
\setcounter{figure}{0}
\setcounter{equation}{0}

\section*{\LARGE Supplementary Material}

\section{Sparsity and computation time\label{sec:compadditional}}

Figure \ref{fig:V2D} further illustrates the sparsity and numerical nonzeros in $\bV$ examined in Figure \ref{fig:Vsparsity}. For all methods except RF-ind (for which $\bW$ and $\bV$ are diagonal), the use of the block forms in \eqref{eq:Urf} or \eqref{eq:Vlf} was very useful in avoiding numerical nonzeros and increasing the sparsity of $\bV$.

\begin{figure}[h!]
	\begin{subfigure}{.19\textwidth}
	\centering
	\includegraphics[width =1\linewidth]{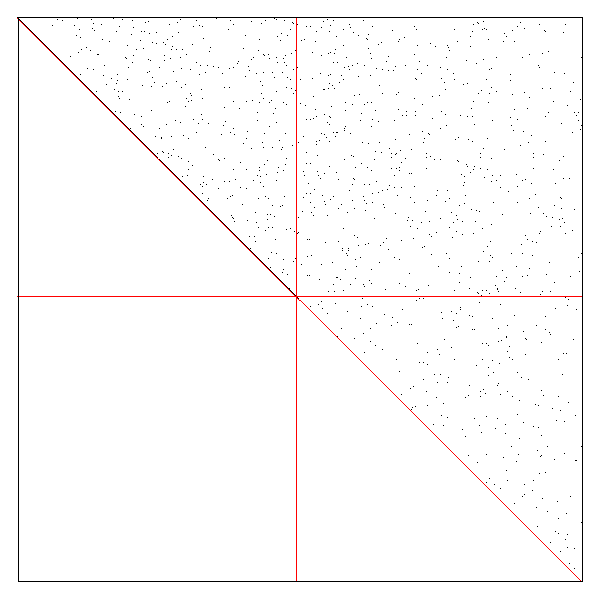}
	\caption{RF-full (1\%)}
	\end{subfigure}%
\hfill
	\begin{subfigure}{.19\textwidth}
	\centering
	\includegraphics[width =1\linewidth]{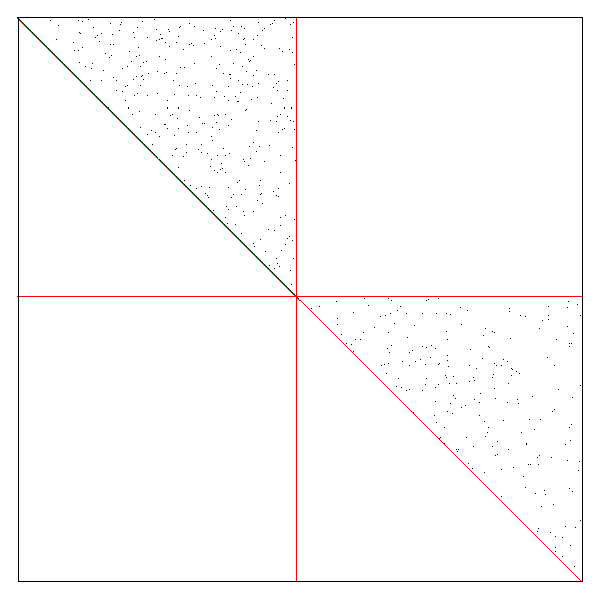}
	\caption{RF-stand (1\%)}
	\end{subfigure}%
\hfill
	\begin{subfigure}{.19\textwidth}
	\centering
	\includegraphics[width =1\linewidth]{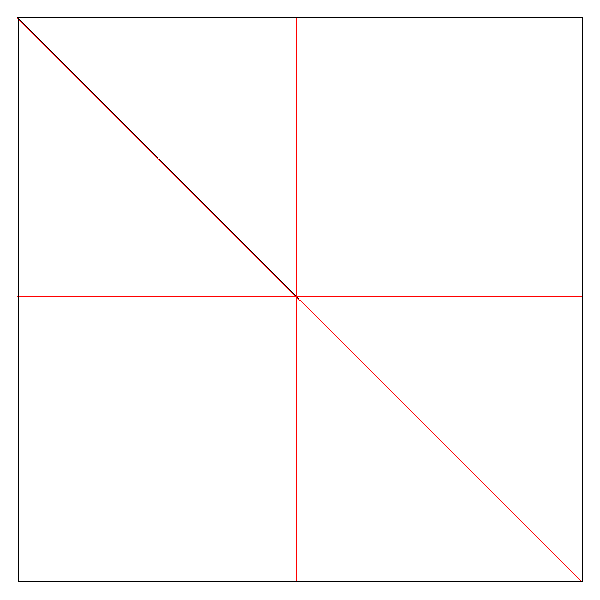}
	\caption{RF-ind (0\%)}
	\end{subfigure}%
\hfill
	\begin{subfigure}{.19\textwidth}
	\centering
	\includegraphics[width =1\linewidth]{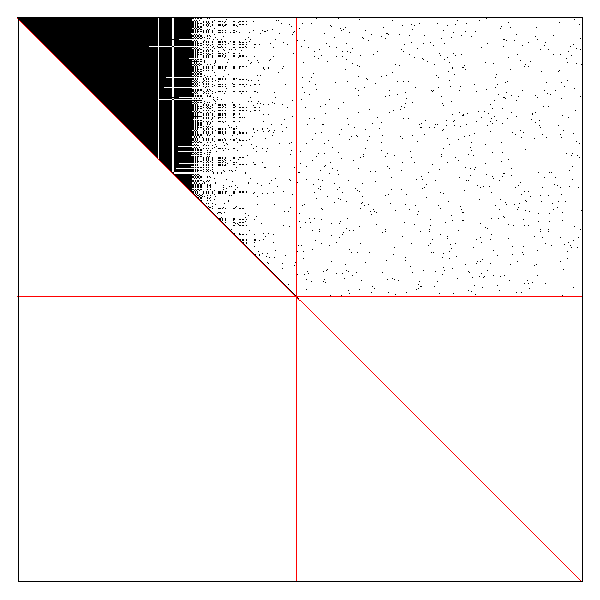}
	\caption{LF-ind (42\%)}
	\end{subfigure}%
\hfill
	\begin{subfigure}{.19\textwidth}
	\centering
	\includegraphics[width =1\linewidth]{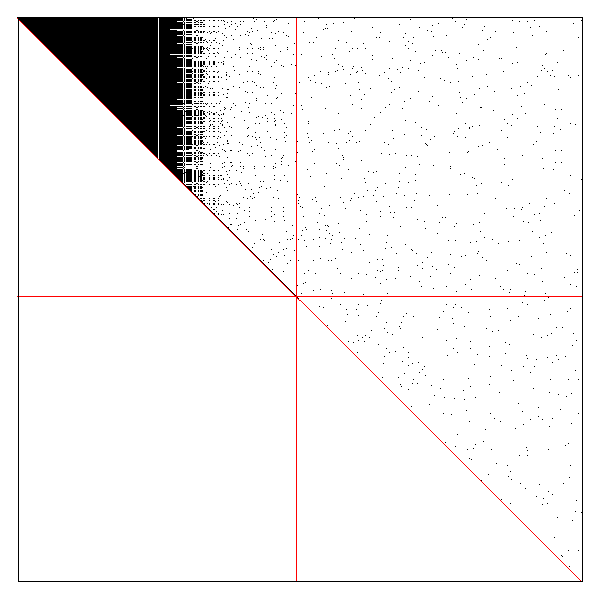}
	\caption{LF-full (40\%)}
	\end{subfigure}

	\begin{subfigure}{.19\textwidth}
	\centering
	\includegraphics[width =1\linewidth]{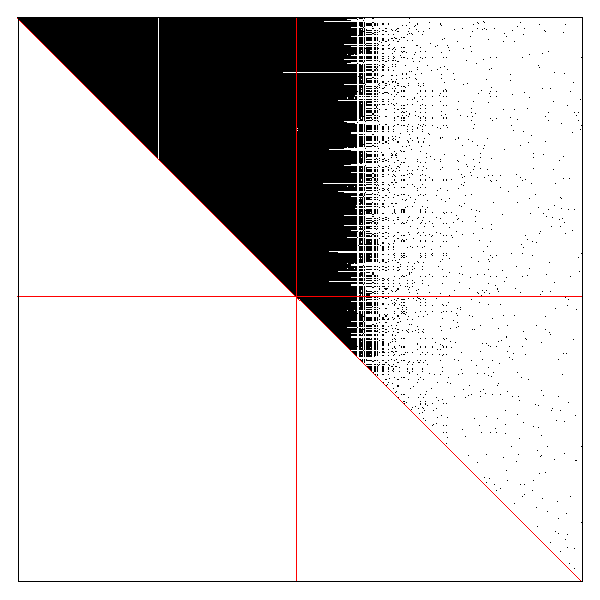}
	\caption{RF-full (100\%)}
	\end{subfigure}%
\hfill
	\begin{subfigure}{.19\textwidth}
	\centering
	\includegraphics[width =1\linewidth]{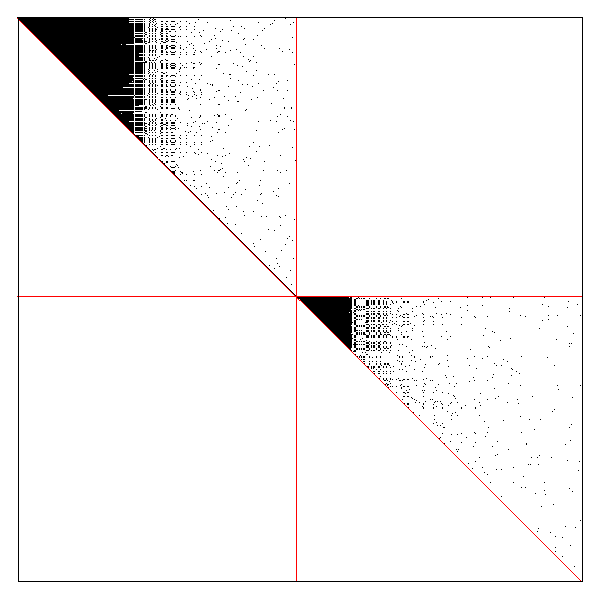}
	\caption{RF-stand (23\%)}
	\end{subfigure}%
\hfill
	\begin{subfigure}{.19\textwidth}
	\centering
	\includegraphics[width =1\linewidth]{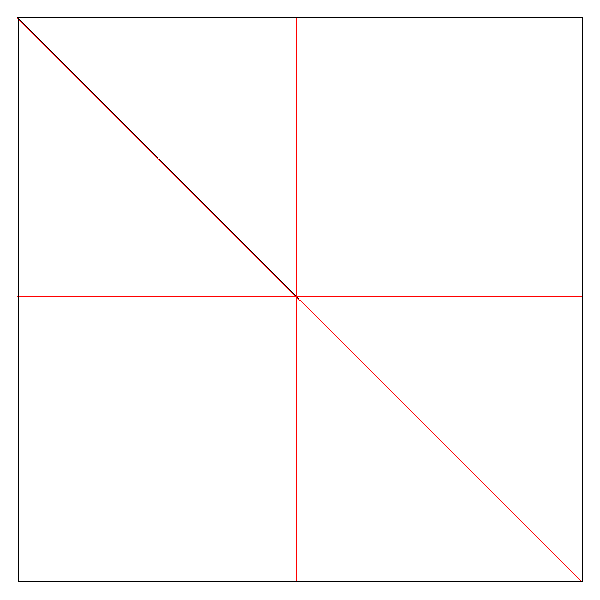}
	\caption{RF-ind (0\%)}
	\end{subfigure}%
\hfill
	\begin{subfigure}{.19\textwidth}
	\centering
	\includegraphics[width =1\linewidth]{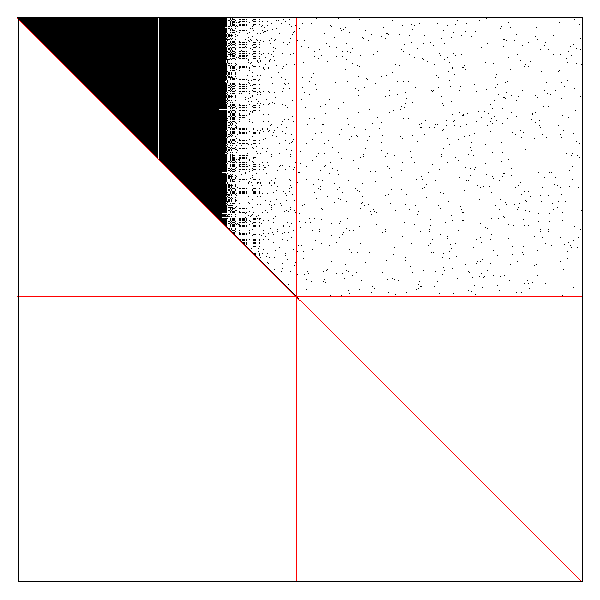}
	\caption{LF-ind (60\%)}
	\end{subfigure}%
\hfill
	\begin{subfigure}{.19\textwidth}
	\centering
	\includegraphics[width =1\linewidth]{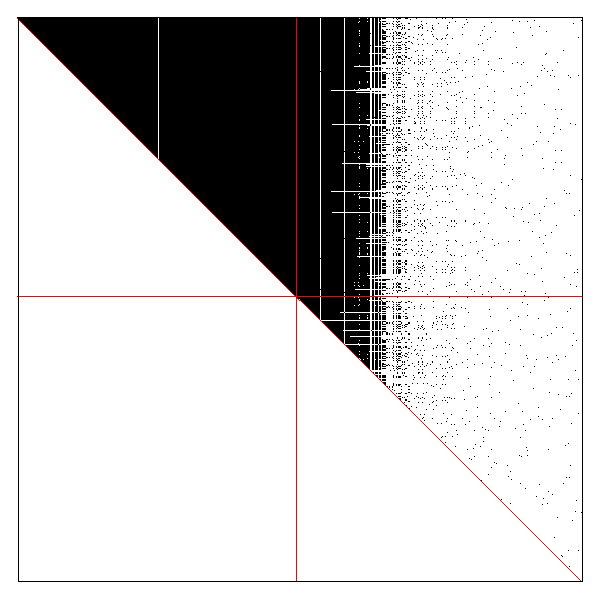}
	\caption{LF-full (100\%)}
	\end{subfigure}%
  \caption{Sparsity structure of $\bV = \rchol(\bW)$ for $n_O=n_P= 1{,}000$ and $m=10$ on a unit square using maxmin ordering. Top row: matrices obtained using \eqref{eq:Urf} or \eqref{eq:Vlf} based on reverse row-column ordering. Bottom row: obtained using ``brute-force'' Cholesky, also based on reverse row-column ordering. Lines separate blocks corresponding to $o$ and $p$. Percentages indicate density (i.e., proportion of nonzero entries) in the upper triangle of $\bV_{oo}$.}
\label{fig:V2D}
\end{figure}

Figure \ref{fig:timing_noperm} shows times for computing $\bU$ and $\bV$. The figure is similar to Figure \ref{fig:timing}, except that for LF-ind and LF-full, $\bV_{oo}$ was computed based on reverse row-column ordering of $\bW_{oo}$ (i.e., without applying a fill-reducing permutation algorithm). Clearly this reverse ordering led to even higher computation times for LF-ind and LF-full, as might have been expected from Figure \ref{fig:Vsparsity}.

\begin{figure}
	\centering
	\begin{subfigure}{.48\textwidth}
	\centering
	\includegraphics[width =.97\linewidth]{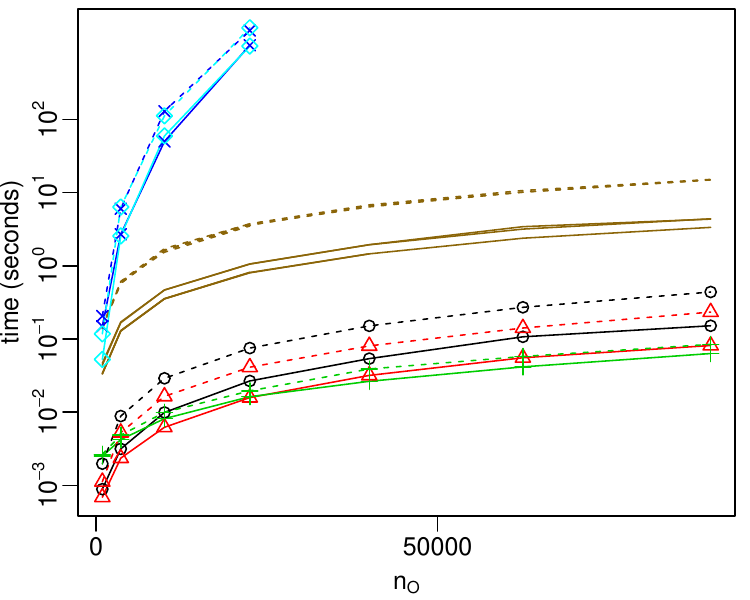}
	\caption{log scale}
	\end{subfigure}%
\hfill
	\begin{subfigure}{.48\textwidth}
	\centering
	\includegraphics[width =.97\linewidth]{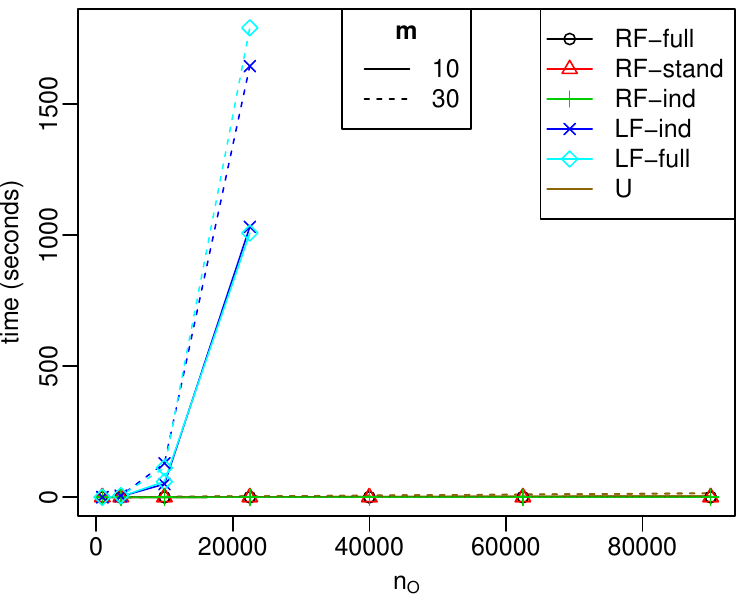}
	\caption{original scale}
	\end{subfigure}
  \caption{Time for computing $\bU$ and $\bV$ for $n_O = n_P$ observed and prediction locations on a unit square, as a function of $n_O$. Time for computing $\bU$ is similar for all methods. For RF methods, time for computing $\bV$ is negligible using \eqref{eq:Urf}. For LF methods, $\bV$ is computed using \eqref{eq:Vlf} and using reverse row-column ordering of $\bW_{oo}$ (i.e., without applying a fill-reducing permutation algorithm)}
\label{fig:timing_noperm}
\end{figure}


\section{Vecchia predictions based on interweaved ordering\label{sec:sgvp}}

We now consider interweaved predictions (IWP), an additional special case of our general Vecchia prediction framework. IWP is based on OP ordering, but the elements of $\by_o$ and $\bz_o$ are interweaved: $\bx = (y_1,z_1,y_2,z_2,\ldots,y_{n_O},z_{n_O},y_{n_O+1},\ldots,y_{n})$, where $\by_p = (y_{n_O+1},\ldots,y_{n})$. The resulting approximation is
\begin{equation*}
\label{eq:iw-b}
\adens(\bx)  = \textstyle\left(\prod_{i=1}^{n} \dens(y_i | \by_{q_y(i)}, \bz_{q_z(i)} )  \right) \left(\prod_{i \in o} \dens(z_i | y_i )\right).
\end{equation*}
IWP tries to maximize latent conditioning, but conditions on $z_j$ instead where necessary to ensure linear complexity. We set $q_y(i) = q(i)$ for $i \in p$. For $i \in o$, IWP follows the same rules as the sparse general Vecchia (SGV) approach \citep{Katzfuss2017a}: Set $q_y(i) \subset q(i)$ such that $j<k$ can only both be in $q_y(i)$ if $j \in q_y(k)$. The remaining conditioning indices are then assigned to $q_z(i) = q(i) \setminus q_y(i)$.
Typically, we determine the set $q_y(i)$ for $i \in o$ similarly to the SGV in \citet{Katzfuss2017a}: For $i=1,\ldots,n_O$, define $h(i) = \argmax_{j \in q(i)} |q_y(j) \cap q(i)|$, $k_i=\argmin_{\ell \in h(i)} \|\bs_i - \bs_\ell\|$, and then set $q_y(i) = (k_i) \cup (q_y(k_i) \cap q(i))$.

IWP is a natural extension of the SGV likelihood approximation in \citet{Katzfuss2017a} to GP predictions. It is easy to show using \eqref{eq:impliedlik} that the likelihood $\adens(\bz_o)$ implied by IWP is equivalent to that of the SGV, and so the two approaches provide a consistent framework for inference and prediction. Because IWP respects OP ordering, we can use \eqref{eq:Vlf} to compute $\bV$ as in the latent-first methods. Specifically, $\bV_{\all,p}$ can simply be copied from $\bW$, and only $\bV_{oo}=\rchol(\bW_{oo})$ must be computed.

This ensures sparsity of $\bV$ for IWP:
\begin{prop}
\label{prop:sgvp}
For IWP, we have $\bV_{ji} = 0$ unless $j=i$ or $j \in q_y(i)$. Hence, $\bV$ has at most $m$ off-diagonal nonzero elements in each column.
\end{prop}
{\footnotesize
\begin{proof}
First, in the graph terminology used in \citet{Katzfuss2017a}, note that any $y_k$ with $k \in p$ cannot have observed descendants under obs--pred ordering. Thus, using Prop.~3.3 in \citet{Katzfuss2017a}, we simply follow the proof of Prop.~6 in \citet{Katzfuss2017a}, considering only the graph formed by $\{y_i: i \in o\}$, to show that $\bV_{ji} = 0$ unless $j=i$ or $j \in q_y(i)$, for $i \in o = (1,\ldots,n_O)$. It is easy to see from the expression \eqref{eq:Vlf} and the definition of $\bU$ in \eqref{eq:U}, that the same holds for the last $n_P$ columns of $\bV$ (i.e., for $i \in p$). This proves that $\bV_{ji} = 0$ unless $j=i$ or $j \in q_y(i)$.

Further, we have $q_y(i) \subset q(i)$, and we have assumed that $|q(i)|\leq m$. Thus, for SGVP, $\bV$ has at most $m$ off-diagonal nonzero entries in each column.
\end{proof}
}
Using the same arguments as in Section \ref{sec:complexity}, this sparsity guarantees that prediction using IWP has linear complexity in $n$.

\section{Numerical comparison to IWP and MRA\label{sec:addcomp}}

Figure \ref{fig:KL2D_alternate} shows a comparison of RF-full to the MRA (Section \ref{sec:latentfirstordering}) and IWP (Section \ref{sec:sgvp}), carried out on a unit square $\domain=[0,1]^2$ with $n_O=n_P=4{,}900$ and effective range $0.15$. All three methods scale linearly in $n$. 

For the MRA, tuning parameters were chosen according to the default settings of the \texttt{GPvecchia} package, which implements the MRA as a special case of the general Vecchia framework \citep[cf.][]{Katzfuss2017a}. Improved accuracy for the MRA could likely be achieved with more careful choice of the tuning parameters. In addition, Figure \ref{fig:KL2D_alternate} does not adjust for the fact that the MRA can be implemented with $\order(nm^2)$ time complexity, which is a factor of $m$ lower than the complexity of RF-full and IWP.

IWP was the most accurate method for small $m$. It is also provably correct (i.e., the KL divergence is zero) for $m=n-1$. However, for intermediate values of $m$ between 10 and 50 or so, the accuracy of IWP did not improve much, and sometimes even decreased. To see why this might be the case, consider a pair $j<i$ with $i, j \in o$, so that $y_j$ and $z_j$ are ordered before $y_i$ in $\bx$. To ensure sparsity using the SGV rules (see Section \ref{sec:sgvp}), $y_i$ might need to condition on $z_j$ instead of $y_j$, which is equivalent to assuming that $y_i$ is conditionally independent of $y_j$ given $z_j$. As a result, the IWP approximation to $\dens(y_i,y_j)$ might be quite poor.

\begin{figure}
\centering\includegraphics[width =.95\linewidth]{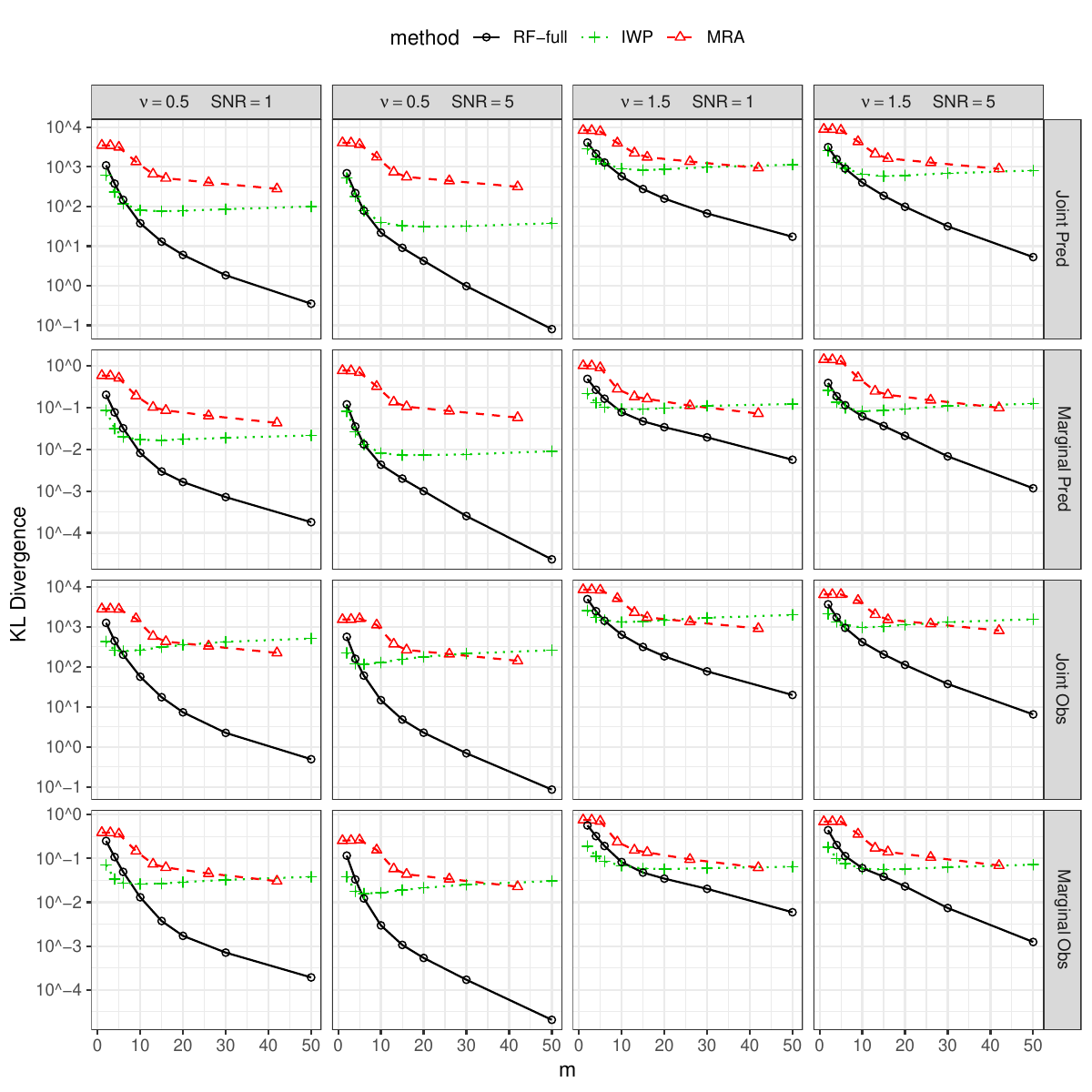}
  \caption{KL divergences (on a log scale) for $\dens(\by_p|\bz_o)$ (Joint Pred), $\dens(\by_{p_i}|\bz_o)$ (Marginal Pred), $\dens(\by_p|\bz_o)$ (Joint Obs), and $\dens(\by_{o_i}|\bz_o)$ (Marginal Obs), based on a GP with Mat\'ern covariance function with smoothness $\nu$, observed with signal-to-noise ratio SNR on a \textbf{unit square}}
\label{fig:KL2D_alternate}
\end{figure}

\section{Satellite data: supplementary plots\label{sec:satsupp}}

Figures \ref{data_and_residuals}--\ref{sd_grid_texas} contain additional plots for our analysis of satellite data in Section \ref{sec:realdata}.

\begin{figure}
\centering
\includegraphics[width=\textwidth]{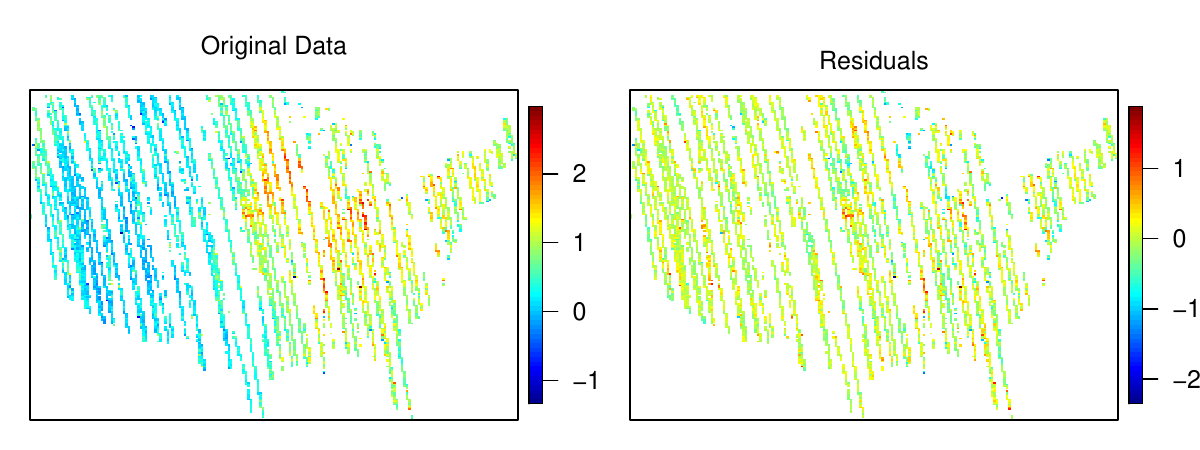}
\caption{\label{data_and_residuals} Original SIF data and residuals after removing Gaussian-basis-function trend}
\end{figure}

\begin{figure}
\centering
\includegraphics[width=\textwidth]{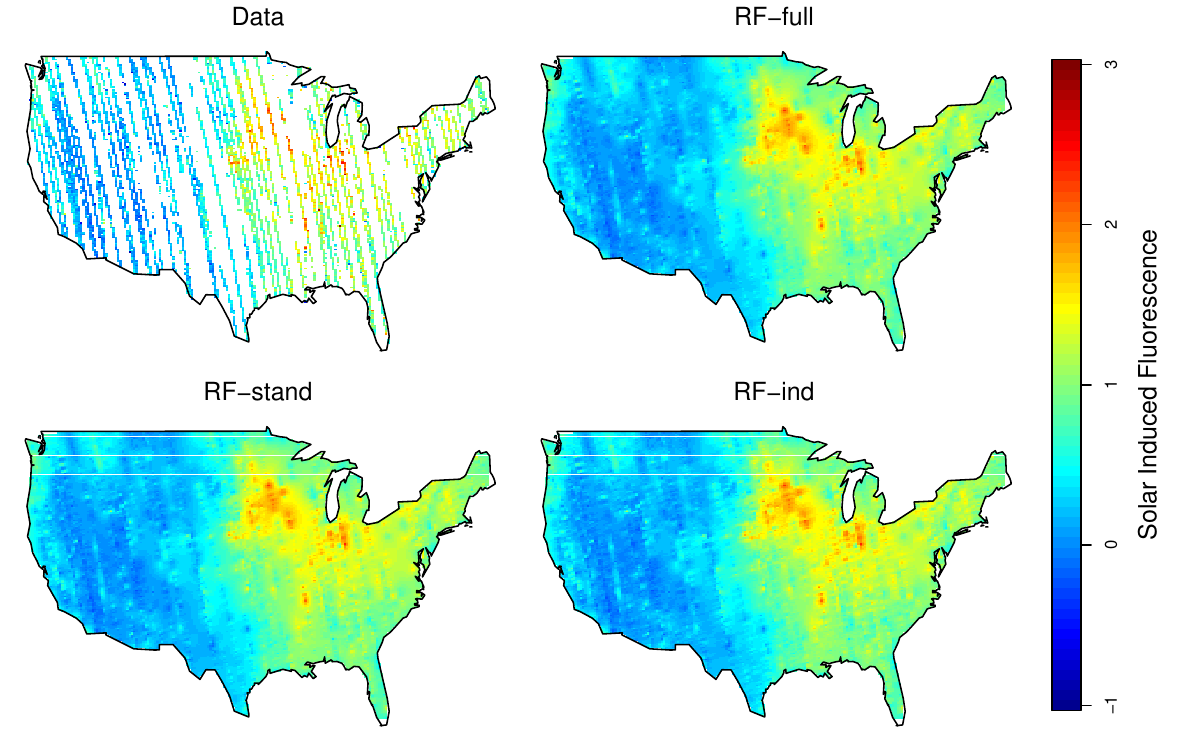}
\caption{\label{pred_grid_usa_m60} Predictions over contiguous USA for $m=60$}
\end{figure}

\begin{figure}
\centering
\includegraphics[width=\textwidth]{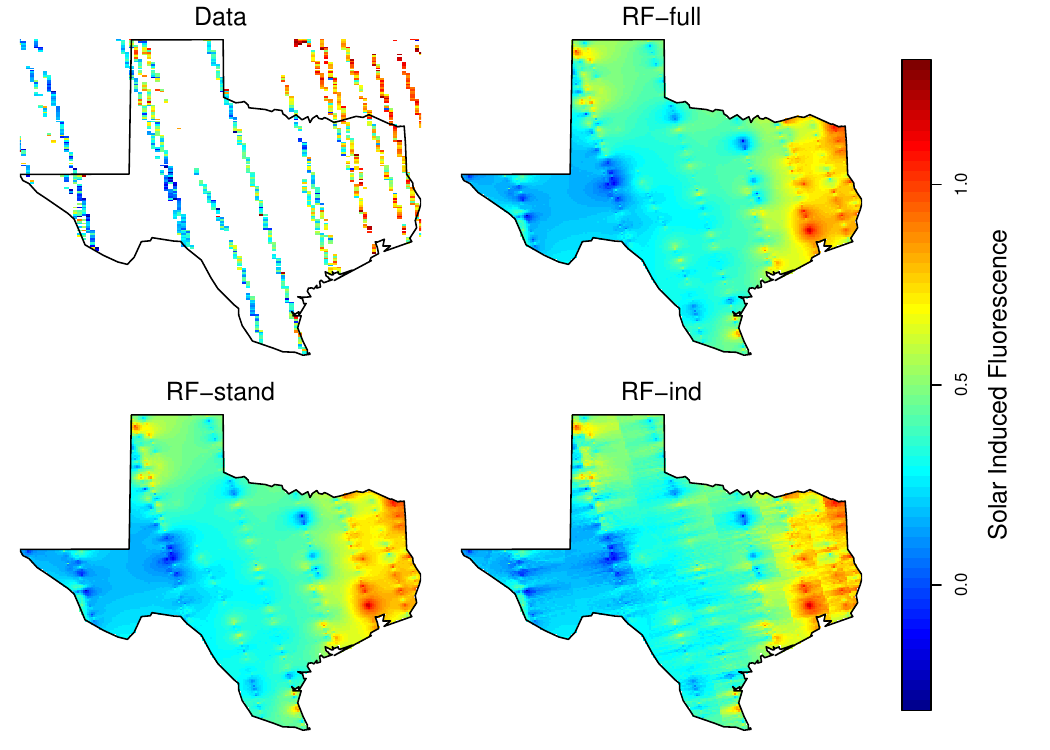}
\caption{\label{pred_grid_texas_m60} Predictions over Texas for $m=60$}
\end{figure}

\begin{figure}
\centering
\includegraphics[width=\textwidth]{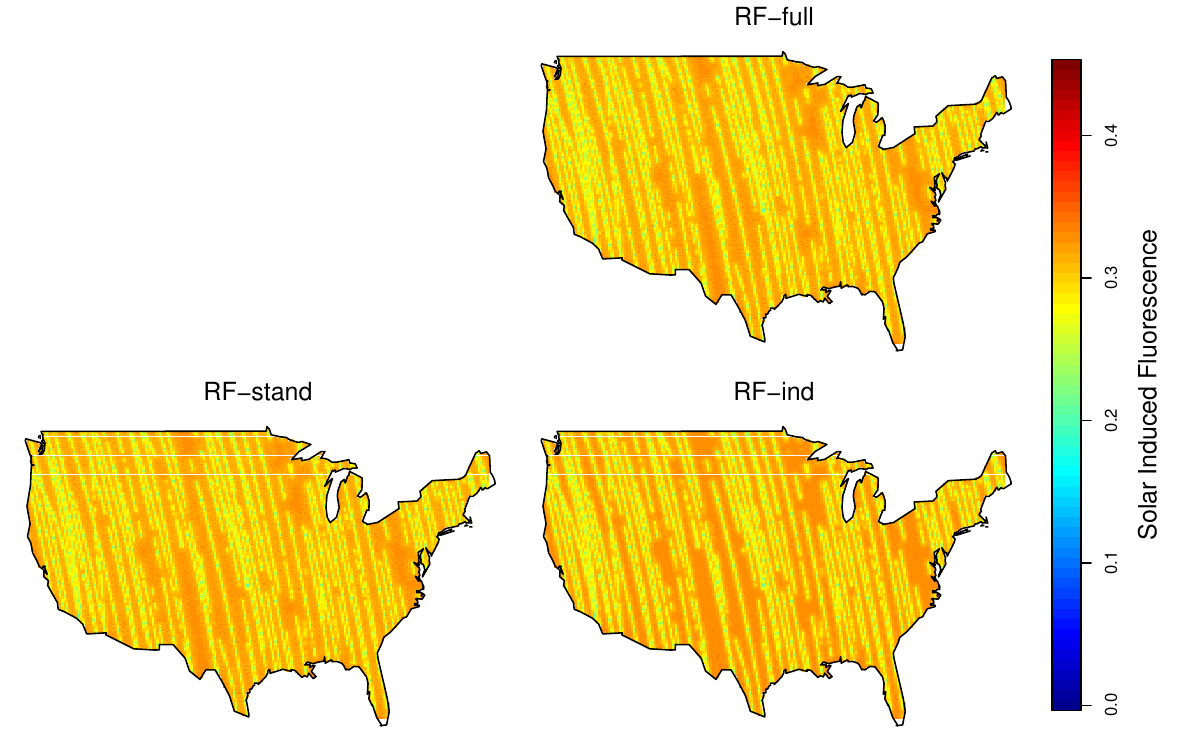}
\caption{\label{sd_grid_usa} Prediction standard deviations over contiguous USA for $m=30$}
\end{figure}

\begin{figure}
\centering
\includegraphics[width=\textwidth]{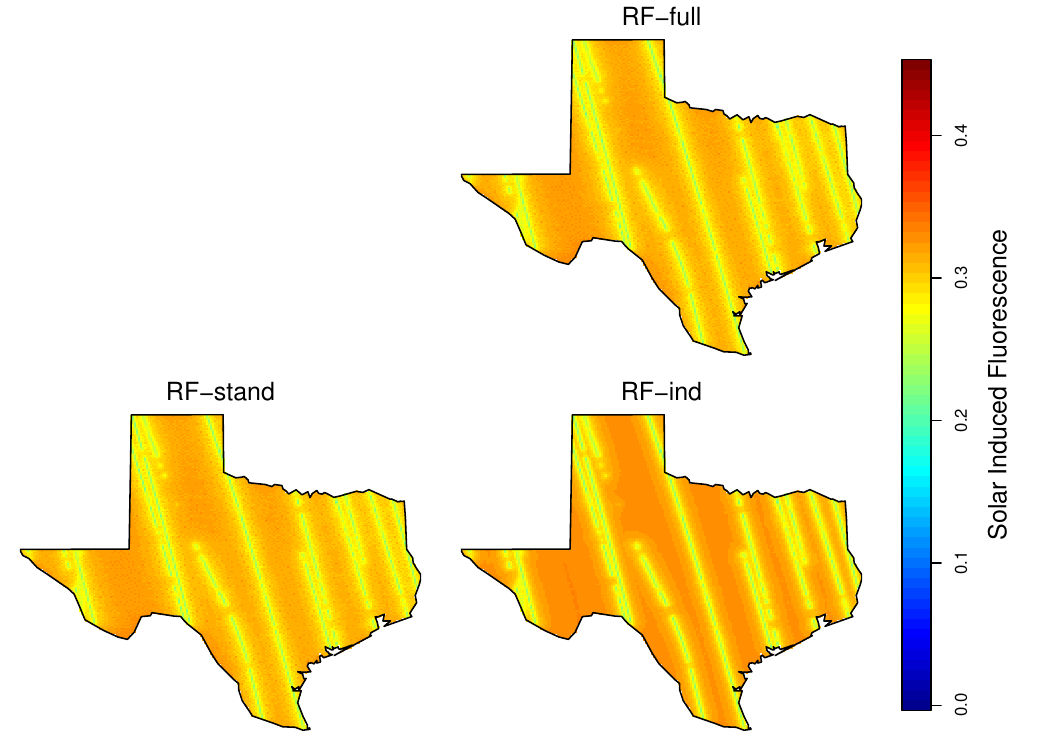}
\caption{\label{sd_grid_texas} Prediction standard deviations over Texas for $m=30$}
\end{figure}


\end{document}